\documentclass[letterpaper,11pt]{article}
\usepackage[letterpaper,margin=1in]{geometry}

\usepackage{mathrsfs}
\usepackage{amsthm}
\usepackage{enumitem}
\usepackage{textgreek}
\usepackage{graphicx}
\usepackage{framed}
\usepackage{booktabs}
\usepackage{tikz-cd}
\usepackage{hyperref}
\usepackage{doi}
\usepackage[framemethod=tikz]{mdframed}
\usepackage{thm-restate}
\usepackage{mathtools}
\usepackage{xspace}
\usepackage[capitalize, nameinlink]{cleveref}
\usepackage{multirow}
\usepackage{cellspace}
\usepackage{array}
\usepackage{tabularx}

\Crefname{remark}{Remark}{Remarks}
\Crefname{observation}{Observation}{Observations}

\theoremstyle{plain}
\newtheorem{theorem}{Theorem}[section]
\newtheorem{lemma}[theorem]{Lemma}

\newtheorem{corollary}[theorem]{Corollary}

\newtheorem{observation}[theorem]{Observation}

\theoremstyle{definition}
\newtheorem{definition}[theorem]{Definition}

\theoremstyle{plain}

\theoremstyle{remark}

\newcommand{\namedref}[2]{\hyperref[#2]{#1~\ref*{#2}}}

\newcommand{\eps}{\varepsilon}

\newcommand{\s}{\mspace{1mu}}
\newcommand{\mybox}[1]{\mspace{2mu}{\setlength{\fboxsep}{1.5pt}\color{lightgray}\boxed{\color{black}\scriptstyle #1}}\mspace{2mu}}
\newcommand{\A}{\mathsf{A}}
\newcommand{\B}{\mathsf{B}}
\newcommand{\C}{\mathsf{C}}
\newcommand{\D}{\mathsf{D}}
\newcommand{\E}{\mathsf{E}}

\renewcommand{\L}{\mathsf{L}}

\newcommand{\Z}{\mathsf{Z}}
\newcommand{\U}{\mathsf{U}}

\newcommand{\M}{\mathsf{M}}
\renewcommand{\P}{\mathsf{P}}
\renewcommand{\O}{\mathsf{O}}

\newcommand{\X}{\mathsf{X}}
\newcommand{\Y}{\mathsf{Y}}

\newcommand{\bX}{\mybox{\X}}

\newcommand{\bMX}{\mybox{\M\X}}

\newcommand{\bOX}{\mybox{\O\X}}

\newcommand{\bPOX}{\mybox{\P\O\X}}
\newcommand{\bMOX}{\mybox{\M\O\X}}
\newcommand{\bMPOX}{\mybox{\M\P\O\X}}
\newcommand{\bZMPOX}{\mybox{\Z\M\P\O\X}}

\newcommand{\ext}{\ensuremath{\operatorname{Ext}}}

\NewCommandCopy{\oldpar}{\paragraph}
\renewcommand{\paragraph}[1]{\oldpar{\normalfont \textbf{#1}}}

\DeclareMathOperator{\re}{\mathcal R}
\DeclareMathOperator{\rere}{\overline{\mathcal R}}

\newenvironment{myabstract}
{\list{}{\listparindent 1.5em%
		\itemindent    \listparindent
		\leftmargin    1cm
		\rightmargin   1cm
		\parsep        0pt}%
	\item\relax}
{\endlist}

\newenvironment{mycover}
{\list{}{\listparindent 0pt
		\itemindent    \listparindent
		\leftmargin    1cm
		\rightmargin   1cm
		\parsep        0pt}%
	\raggedright
	\item\relax}
{\endlist}

\newcommand{\myemail}[1]{\,$\cdot$\, {\small #1}}
\newcommand{\myaff}[1]{\,$\cdot$\, {\small #1}\par\smallskip}

\definecolor{darkgreen}{rgb}{0,0.5,0}
\definecolor{darkred}{rgb}{0.4,0,0}
\hypersetup{
	colorlinks=true,
	linkcolor=darkred,
	citecolor=darkgreen,
	filecolor=black,
	urlcolor=[rgb]{0,0.1,0.5},
	pdftitle={FIXME},
	pdfauthor={FIXME}
}

\newcommand\blfootnote[1]{
    \begingroup
    \renewcommand\thefootnote{}\footnote{#1}
    \addtocounter{footnote}{-1}
    \endgroup
}

\begin{document}

\begin{mycover}
    {\huge\bfseries\centering Tight Lower Bounds in the\\Supported LOCAL Model
    \par}
  \bigskip
  \bigskip
  \bigskip

  \textbf{Alkida Balliu}
  \myemail{alkida.balliu@gssi.it}
  \myaff{Gran Sasso Science Institute, L'Aquila, Italy}
  
  \textbf{Thomas Boudier}
  \myemail{thomas.boudier@gssi.it}
  \myaff{Gran Sasso Science Institute, L'Aquila, Italy}
  
  \textbf{Sebastian Brandt}
  \myemail{brandt@cispa.de}
  \myaff{CISPA Helmholtz Center for Information Security, Saarbr\"ucken, Germany}
  
  \textbf{Dennis Olivetti}
  \myemail{dennis.olivetti@gssi.it}
  \myaff{Gran Sasso Science Institute, L'Aquila, Italy}

\blfootnote{This work has been partially funded by the PNRR MIUR research project GAMING ``Graph Algorithms and MinINg for Green agents'' (PE0000013, CUP D13C24000430001), and by the research project RASTA ``Realtà Aumentata e Story-Telling Automatizzato per la valorizzazione di Beni Culturali ed Itinerari'' (Italian MUR PON Project ARS01 00540).}
  
\end{mycover}
\bigskip

\begin{myabstract}
    In this work, we study the complexity of fundamental distributed graph problems in the recently popular setting where information about the input graph is available to the nodes before the start of the computation.
    We focus on the most common such setting, known as the Supported LOCAL model, where the input graph---on which the studied graph problem has to be solved---is guaranteed to be a subgraph of the underlying communication network.

    Building on a successful lower bound technique for the LOCAL model called round elimination, we develop a framework for proving complexity lower bounds in the stronger Supported LOCAL model.
    Our framework reduces the task of proving a (deterministic or randomized) lower bound for a given problem $\Pi$ to the graph-theoretic task of proving non-existence of a solution to another problem $\Pi'$ (on a suitable graph) that can be derived from $\Pi$ in a mechanical manner.

    We use the developed framework to obtain substantial---and, in the majority of cases, asymptotically tight---Supported LOCAL lower bounds for a variety of fundamental graph problems, including maximal matching, maximal independent set, ruling sets, arbdefective colorings, and generalizations thereof.
    In a nutshell, for essentially any major lower bound proved in the LOCAL model in recent years, we prove a similar lower bound in the Supported LOCAL model.

    Our framework also gives rise to a new \emph{deterministic} version of round elimination in the LOCAL model: while, previous to our work, the general round elimination technique required the use of randomness (even for obtaining deterministic lower bounds), our framework allows to obtain deterministic (and therefore via known lifting techniques also randomized) lower bounds in a purely deterministic manner.
    Previously, such a purely deterministic application of round elimination was only known for the specific problem of sinkless orientation [SOSA'23].

\end{myabstract}

\clearpage
\tableofcontents
\clearpage
\section{Introduction}
Since its beginning, one of the cornerstones of research in distributed computation has been the study of locality, asking how distant information the nodes of a large computer network need to collect in order to solve a given computational problem.
While this fundamental question has been commonly studied in the idealized setting where nothing about the input instance is known before the start of the computation, in this paper we will consider a more general setting: what if information about the instance is known ahead of time, allowing for some preprocessing before the start of the actual computation?
Besides being a highly natural question from a theoretical perspective that has been studied in many areas of computer science, this question has a very concrete motivation in the context of distributed algorithms, making its study particularly relevant in this field.

In distributed computing, the common scenario is that of a network on which computational problems, often related to the structure of the network, have to be solved.
This network is usually assumed to be \emph{static}, not allowing for any (or only few) changes in the topology.
In contrast, on this fixed network, in many practical settings we want to solve \emph{a variety} of computational tasks that occur \emph{repeatedly} over time, often defined on subnetworks of the entire network (e.g., involving different subsets of the computational entities the network consists of).
Naturally, in such a setting, it is computationally desirable to exploit the consistency of the network topology by performing some preprocessing on the network itself \emph{once} and store the computed information in the nodes of the network to allow for a more rapid execution of the computational tasks that occur over time.

\paragraph{Formalization.}
While formalized already in 2013 (in the context of software defined networks) by Schmid and Suomela~\cite{schmidexploiting13} as so-called ``supported models'', the study of this highly natural setting has received increasing attention in the last five years~\cite{foersterpower19,foersterpreprocessing19,guptasparse22,supportedopodis,BalliuKKLOPPR0S23,HaeuplerWZ21}, covering ``supported'' versions of the LOCAL model, the CONGEST model, and others.
The most frequently studied model, commonly known under the name \emph{Supported LOCAL}, captures the setting discussed above as follows.
The input instance is given by a communication network $G$, called the \emph{support graph}, and a subgraph $G'$ of $G$, called the \emph{input graph}, on which some given problem has to be solved.
Each node has a unique identifier and is aware of which of its incident edges are part of $G'$, if any; otherwise it has no information about $G'$.
However, each node is aware of the entire communication network $G$ (i.e., each node has complete information about $G$'s topology, all assigned unique identifiers etc.).
The actual computation proceeds as in the standard LOCAL model of distributed computation on the communication network $G$.
For a formal introduction to the (Supported) LOCAL model, we refer the reader to \Cref{sec:preliminaries}.

We note that, in Supported LOCAL, the aforementioned preprocessing is modeled in a very powerful way: each node has \emph{complete} information about $G$ as opposed to partial information obtained by some actual preprocessing step.
However, this makes our main results only stronger as they are (round complexity) \emph{lower} bounds.

\paragraph{Significance for lower bounds.}
Besides its importance in capturing a natural and frequent setting, the Supported LOCAL model turns out to be also highly relevant for lower bounds in models related to locality (such as the LOCAL model), as we will explain in the following.

In the LOCAL model, recent years have seen a revolution with regards to round complexity lower bounds.
While 8 years ago only a handful of non-trivial lower bounds were known for the complexity of reasonably important problems, the introduction of a new lower bound technique (in a version tailored towards a specific problem~\cite{brandtlower16} in 2016, and in its general version~\cite{brandtautomatic19} in 2019) has advanced the state of the art for lower bounds dramatically.
Using this technique, called round elimination, a series of recent works has established substantial complexity lower bounds for many of the most important problems studied in the context of locality, such as sinkless orientation \cite{brandtlower16}, maximal matching and variants thereof \cite{Balliu2019, trulytight}, maximal independent sets and many variants of it, such as ruling sets and bounded-outdegree dominating sets \cite{rs-siam,outdegree-ds,Balliu0KO22}, arbdefective colorings \cite{Balliu0KO22}, and even variants of these problems on hypergraphs \cite{mm-hypergraphs}. 
While elegant and powerful, one surprising aspect of the round elimination technique is that, even for obtaining deterministic lower bounds, it is currently required to use randomness.
More specifically, to obtain such a deterministic bound with round elimination, first a randomized lower bound is proved, and then the randomized bound is lifted to a deterministic lower bound.
For an overview of the whole round elimination framework, including the final lifting of the obtained bound, we refer the reader to \cite{Balliu0KO22}.
While the obtained deterministic lower bounds are generally as good as can be expected (even if using randomness were not required), the requirement to first obtain a randomized lower bound is highly unsatisfactory, due to the following reasons.

First of all, while round elimination currently applies directly essentially only to the LOCAL model (although the obtained lower bounds indirectly imply, for instance, the same bounds for strictly weaker models such as the CONGEST model), a highly interesting direction of research is to extend this technique also to other models related to the notion of locality (such as SLOCAL, LCA, and VOLUME).\footnote{While such an extension to the Supported LOCAL model was already shown in \cite{BalliuKKLOPPR0S23} for a specific problem, the paper at hand is a first example of a problem-independent extension of round elimination to the Supported LOCAL model.}
It is far from clear whether, for such extensions, the roundabout way of first obtaining and then lifting randomized lower bounds is feasible (and provides enough power to obtain tight bounds), while an approach not involving randomness might be more promising (and does not risk degradation of the size of the bound).
Second, from the viewpoint of simplicity, a direct way of proving deterministic lower bounds via round elimination \emph{deterministically} that avoids cumbersome failure probability analyses and the application of a blackbox lifting theorem would be highly desirable.
We remark that this would essentially also remove the need to use randomness in round elimination for obtaining \emph{randomized} lower bounds: the aforementioned deterministic lower bounds directly imply randomized lower bounds (that are as good as currently achievable via ``randomized'' round elimination) via the interesting fact that for all problems to which round elimination is applicable, the randomized complexity on $n$-node graphs is at least the deterministic complexity on roughly $(\log n)$-sized instances \cite{ChangKP19,derandomization}.
Third, the sizes of the bounds obtainable with the current version of round elimination directly depend on the number of labels used in the descriptions of certain problems generated in a mechanical manner from the problem of interest (see~\cite[Theorem 7.1]{Balliu0KO22}). This is an inherent consequence of the randomized analysis used in the current approach and can lead to worse bounds than achievable with an approach that avoids randomness.

Interestingly, a direct deterministic way to obtain a deterministic (and, by implication, a randomized) lower bound via round elimination was recently shown for the sinkless orientation problem~\cite{BalliuKKLOPPR0S23}---via the Supported LOCAL model.\footnote{We note that the idea of providing a support graph as input to the nodes is highly related to the \emph{ID graph} technique from~\cite{brandt2021local} that was (independently) developed in the context of a different lower bound technique called \emph{Marks' technique}.}
While this was done in a highly problem-specific manner, this raises the question whether such an approach is feasible also for other problems and, more generally, whether a general \emph{deterministic} round elimination framework can be developed.
We remark that sinkless orientation is a problem that behaves very nicely under round elimination, which is likely the reason why round-elimination-related results are likely to be obtained first for this problem (see, e.g., \cite{brandtlower16})---obtaining similar results for problems with a much more complex behavior under round elimination has historically required the development of more generally applicable techniques~\cite{brandtautomatic19}.

\subsection{Our Contributions}
\paragraph{A deterministic round elimination framework and Supported LOCAL lower bounds.}
As one of our main conceptual contributions, building on the original round elimination framework from~\cite{brandtautomatic19,Balliu2019}, we develop a \emph{deterministic} round elimination framework that avoids using randomness.
To a large degree, this framework is based on a new technique for proving lower bounds in the Supported LOCAL model that is of independent interest.

We start our discussion by giving a brief overview of the current version of round elimination.
In a nutshell, in order to prove a lower bound for a given problem $\Pi$ in the LOCAL model, the original round elimination framework provides a simple blueprint: first derive a sequence of problems $\Pi = \Pi_0, \Pi_1, \Pi_2, \dots$ from $\Pi$ in a well-defined mechanical manner, and then show that a problem $\Pi_k$ with (ideally large) index $k$ cannot be solved in $0$ rounds with a randomized algorithm that allows only a certain failure probability.
From this, a randomized lower bound can be inferred (whose size depends on the size of the index $k$), which subsequently can be lifted to a deterministic lower bound via a blackbox lifting theorem \cite{Balliu0KO22}.
This framework is applicable to any problem $\Pi$ that is \emph{locally checkable}, which, roughly speaking, means that $\Pi$ can be described via local constraints (for a formal definition, see \Cref{sec:preliminaries}).
The class of locally checkable problems contains the vast majority of problems studied in the LOCAL model, including essentially all common colorings problems, maximal matching, maximal independent set, and many more.

Our first technical contribution consists in showing that, by replacing the LOCAL model with the Supported LOCAL model in the above blueprint, we can avoid the use of randomness.
More precisely, we prove that the fact that problem $\Pi_k$ from the aforementioned sequence cannot be solved in $0$ rounds by a \emph{deterministic} algorithm in the \emph{Supported LOCAL model} implies a \emph{deterministic} lower bound for $\Pi$ (whose size depends on $k$).
The following theorem, proved in \Cref{sec:re}, formalizes this statement (where $\Delta$ is the maximum degree of the support graph, and $n$ the number of nodes).

\begin{theorem}[Simplified version of \Cref{lem:re-works}]\label{lem:re-to-bounds-intro}
    Assume there is no deterministic $0$-round algorithm for $\Pi_k$ in the Supported LOCAL model.
    Then, any deterministic algorithm solving $\Pi$ in the Supported LOCAL model requires $\Omega(\min\{k,\log_{\Delta} n\})$ rounds.
\end{theorem}

\Cref{lem:re-to-bounds-intro} generalizes the special case of this theorem that was proved for the sinkless orientation problem in~\cite{BalliuKKLOPPR0S23}. 

In order to make use of \Cref{lem:re-to-bounds-intro} for proving lower bounds in the Supported LOCAL model and, by consequence, developing our deterministic round elimination framework for the LOCAL model, we provide a characterization of deterministic $0$-round-solvability in the Supported LOCAL model.
More precisely, in \Cref{sec:new-technique}, for any locally checkable problem $\Psi$, we show how to define a problem $\mathrm{lift}(\Psi)$ satisfying the following surprising property.
\begin{theorem}[Simplified version of \Cref{th:lift-bipartite}]\label{th:lift-intro}
    Problem $\Psi$ can be solved in $0$ rounds in the Supported LOCAL model on a support graph $G$ if and only if there \emph{exists} a solution on $G$ for $\mathrm{lift}(\Psi)$.
\end{theorem}
By combining \Cref{lem:re-to-bounds-intro} and \Cref{th:lift-intro} (where we set $\Psi := \Pi_k$), the task of proving lower bounds in the Supported LOCAL model for some problem $\Pi$ is reduced to the task of showing that another problem $\Pi'$ admits no solution in the support graph(s). This dramatically simplifies the task of proving Supported LOCAL lower bounds, since we do not need to think about distributed algorithms executed on subgraphs, but instead can obtain lower bounds by answering purely graph-theoretic questions.

Besides providing a highly useful technique for proving lower bounds in the Supported LOCAL model, the combination of \Cref{lem:re-to-bounds-intro} and \Cref{th:lift-intro} also provides the desired deterministic round elimination framework (for the LOCAL model), as lower bounds in the stronger Supported LOCAL model immediately apply also to the weaker LOCAL model.

While all of the aforementioned lower bounds are deterministic, in \Cref{sec:rand} we provide a lifting theorem for the Supported LOCAL model that allows us to turn deterministic lower bounds into randomized lower bounds.

\begin{theorem}[Simplified version of \Cref{lem:derand-supported}]\label{lem:rand-bounds-intro}
    Let $D_\Pi(n)$ denote the deterministic complexity of $\Pi$ in the Supported LOCAL model, and $R_\Pi(n)$ the randomized complexity of $\Pi$ in the Supported LOCAL model.
    Then, 
    \[
        D_\Pi(n) \le R_\Pi(2^{O(n^2)}).
    \]
\end{theorem}

\Cref{lem:rand-bounds-intro} constitutes an extension of the celebrated lifting theorem from~\cite{ChangKP19,derandomization} to the Supported LOCAL model.

Using \Cref{lem:rand-bounds-intro}, we can obtain randomized lower bounds (in Supported LOCAL, and therefore also LOCAL) from our deterministic framework.
More precisely, combining \Cref{lem:re-to-bounds-intro,th:lift-intro,lem:rand-bounds-intro}, we obtain the following theorem, summarizing the above discussion.

\begin{theorem}[Simplified version of \Cref{thm:approach-bipartite}]\label{thm:approach-intro}
    Suppose there exists a support graph $G$ on which no solution for $\mathrm{lift}(\Pi_k)$ exists.
    Then, any deterministic algorithm solving $\Pi$ on $G$ requires $\Omega(\min\{k,\log_{\Delta} n\})$ rounds and any randomized algorithm solving $\Pi$ requires $\Omega(\min\{k,\log_{\Delta} \log n\})$ rounds in the Supported LOCAL model.
\end{theorem}

We remark that \Cref{lem:re-to-bounds-intro,th:lift-intro,lem:rand-bounds-intro,thm:approach-intro} are simplified versions of the actual theorems we prove.
In particular, we obtain the theorems in a more general setting, covering also hypergraphs and bipartite graphs.

After giving an overview of our contributions regarding the development of new techniques, we now turn our focus to results for concrete (classes of) problems.
In a nutshell, for essentially all major lower bounds proved in the LOCAL model in recent years \cite{Balliu2019, trulytight, outdegree-ds, rs-siam,Balliu0KO22}, we obtain similar lower bounds in the stronger Supported LOCAL model, including lower bounds for maximal matching, maximal independent set, ruling sets, arbdefective coloring, and generalizations of these problems, providing ample evidence for the usefulness of our new technique.
However, we would like to emphasize that, while this technique provides a clear blueprint for obtaining Supported LOCAL lower bounds, each such lower bound still requires proving an existential graph-theoretic statement, which is far from trivial.

\paragraph{Maximal matching and variants.}
An $x$-maximal $y$-matching of a graph $G =(V,E)$ is a subset $M \subseteq E$ of edges satisfying that each node is incident to at most $y$ edges of $M$, and that, if a node $v$ is not incident to any edge of $M$, then at least $\min\{\deg(v),\Delta-x\}$ neighbors of $v$ are incident to an edge of $M$. This family of problems includes maximal matching (by setting $x=0$ and $y=1$), but it also includes many variants of it, for example, the problem of computing some relaxed variant of matching, where nodes are allowed to be matched multiple times, say $\log \Delta$, and unmatched nodes need to have at least $\Delta/2$ matched neighbors.
This class of problems was studied in \cite{Balliu2019}, where non-tight bounds were shown in the LOCAL model. Later, \cite{trulytight} provided tight bounds, for any values $x$ and $y$. 
The following theorem shows that the same bounds hold in the stronger Supported LOCAL model as well. For $\Delta' = \Theta(\Delta)$, this bound matches exactly the one known for LOCAL. This theorem is proved in \Cref{sec:matching}.

\begin{theorem}[Simplified version of \Cref{thm:lb-xy-matching}]\label{thm:lb-xy-matching-intro}
    On bipartite $2$-colored support graphs of degree $\Delta$, when the input graph has degree $\Delta'$, the $x$-maximal $y$-matching problem requires $\Omega(\min\{(\Delta'-x) / y,\log_\Delta n\})$ rounds deterministically and $\Omega(\min\{(\Delta'-x) / y,\log_\Delta \log n\})$ with randomization.
\end{theorem}
\noindent In \cite{supportedopodis} the authors proved that if maximal matching on bipartite $2$-colored graphs can be solved in $o(\Delta')$ rounds, then it can be solved in $o(\Delta')$ rounds on all graphs. They left as an open question to determine the complexity of maximal matching on $2$-colored graphs. Observe that \Cref{thm:lb-xy-matching-intro} solves the open question of \cite{supportedopodis} by answering it negatively.

\paragraph{Arbdefective coloring.}
The $\alpha$-arbdefective $c$-coloring problem requires to color the nodes with $c$ colors, and output an orientation of the edges that connect nodes of the same color, such that each node has at most $\alpha$ outgoing edges. Arbfective colorings are a basic building block that has been extensively used to develop many algorithms for proper coloring \cite{MausTonoyan20,fraigniaud16local,barenboim16sublinear}.
It is known that in the LOCAL model, this problem can be solved in $O(\Delta + \log^* n)$ rounds when $(\alpha+1)c > \Delta$, while it requires $\Omega(\log_\Delta n)$ for deterministic algorithms and $\Omega(\log_\Delta \log n)$ for randomized ones when $(\alpha+1)c \le \Delta$ \cite{Balliu0KO22}.
The following theorem, proved in \Cref{sec:arb}, shows that a similar statement holds in the Supported LOCAL model as well, where we want to compute an $\alpha$-arbdefective $c$-coloring problem on a given input graph $G'$ of degree $\Delta'$.
\begin{theorem}[Simplified version of \Cref{thm:arb-coloring}]\label{thm:arb-coloring-intro}
    On support graphs of degree $\Delta$, when the input graph has degree $\Delta'$, if $(\alpha+1)c \le \min\{\Delta', \epsilon \Delta / \log \Delta \}$ for a small-enough constant $\epsilon$, then the $\alpha$-arbdefective $c$-coloring problem requires $\Omega(\log_\Delta n)$ deterministic rounds and $\Omega(\log_\Delta \log n)$ randomized rounds.
\end{theorem}

\paragraph{Arbdefective colored ruling sets.}
A subset $S \subseteq V$ of nodes is an $\alpha$-arbdefective $c$-colored $\beta$-ruling set if the following holds: the subgraph induced by nodes of $S$ is labeled with a solution for the $\alpha$-arbdefective $c$-coloring problem, and for each node $v \in V \setminus S$, it holds that there is a node in $S$ within distance $\beta$.

The $\alpha$-arbdefective $c$-colored $\beta$-ruling sets problem family includes many natural and widely studied problems. For example, by setting $\beta=0$, we obtain $\alpha$-arbdefective $c$-coloring, by then setting $\alpha =0$ we obtain $c$-coloring.
If we set $\alpha = 0$ and $c=1$, we obtain $(2,\beta)$-ruling sets, and by then setting $\beta=1$ we obtain the maximal independent set problem. Instead, if we set $\beta = 0$, $\alpha = \Delta - 1$, and $c = 1$, we obtain sinkless coloring (which is, up to one round of communication, equivalent to sinkless orientation) \cite{brandtlower16}, while if we set $\beta = 1$ and $c = 1$, we obtain $\alpha$-outdegree dominating sets \cite{outdegree-ds}.
Hence, lower bounds for this general problem family are highly useful as they imply lower bounds for all the important problems that are special cases of this family.

In \cite{Balliu0KO22}, it has been shown that computing an $\alpha$-arbdefective $c$-colored $\beta$-ruling set in the LOCAL model requires $T_{\mathrm{det}} = \Omega(\min\{\beta(\frac{\Delta}{(\alpha + 1)c})^{1 / \beta}), \log_\Delta n \})$ rounds for deterministic algorithms and $T_{\mathrm{rand}} = \Omega(\min\{\beta(\frac{\Delta}{(\alpha + 1)c})^{1 / \beta}), \log_\Delta \log n \})$ rounds for randomized ones, assuming that $\beta$ is small enough to imply $T_{\mathrm{det}} - \beta = \Omega(T_{\mathrm{det}})$ in the deterministic case, and $T_{\mathrm{rand}} - \beta = \Omega(T_{\mathrm{rand}})$ in the randomized case.
This result is tight: in fact, these lower bounds hold even if a $(\Delta+1)$-coloring is provided to the nodes, and in such a setting it is possible to compute an  $\alpha$-arbdefective $c$-colored $\beta$-ruling set  in $O(\beta(\frac{\Delta}{(\alpha + 1)c})^{1 / \beta}))$ rounds.

In \Cref{sec:acbrs}, we show that similar lower bounds hold in the Supported LOCAL model as well. More in detail, we focus on the case $\beta \ge 1$ (since the case $\beta = 0$ is already handled by \Cref{thm:arb-coloring-intro}), and prove the following theorem, which, for $\beta = O(1)$ and $\Delta$ sufficiently large compared to $\Delta'$, matches the lower bounds known in the LOCAL model.
\begin{theorem}[Simplified version of \Cref{thm:acbrs}]\label{thm:acbrs-intro}
Let $\alpha\ge 0$, $c \ge 1$, $\beta \ge 1$ be integers, and assume $\beta = O(1)$ is a constant, i.e., it does not depend on $n$ and $\Delta$. 
Let $\Delta$ be the degree of the support graph, and $\Delta'$ be the degree of the input graph.
Let $\Bar{\Delta} = \min\{\Delta',  \Delta / \log \Delta \}$.
Then, the $\alpha$-arbdefective $c$-colored $\beta$-ruling set problem requires $\Omega(\min\{ (\frac{\Bar{\Delta}}{(\alpha + 1)c})^{1 / \beta}), \log_\Delta n \})$ deterministic rounds and $\Omega(\min\{ (\frac{\Bar{\Delta}}{(\alpha + 1)c})^{1 / \beta}), \log_\Delta \log n \})$ randomized rounds.
\end{theorem}
We note that, in \Cref{thm:arb-coloring-intro}, the $\Delta / \log \Delta$ term is necessary: the support graphs for which the theorems are proved satisfy that their chromatic number is upper bounded by $O(\Delta / \log \Delta)$, and hence the nodes can compute an $O(\Delta/ \log \Delta)$-coloring without communication. Similarly, it is known that, given a $k$-coloring, one can compute an $\alpha$-arbdefective $c$-colored $\beta$-ruling set in $O(\frac{k}{(\alpha + 1)c})^{1 / \beta})$ rounds \cite{Balliu0KO22}, and hence the term $\Delta / \log \Delta$ is required also in \Cref{thm:acbrs-intro}.

In \cite{supportedopodis} the authors noted that MIS can be solved in $\chi_G$ rounds, where $\chi_G$ 
is the chromatic number of the support graph $G$. As an open question, they asked whether this bound could be improved. \Cref{thm:acbrs-intro} implies that this is not possible, at least for deterministic algorithms. In fact, if we set $\Delta := \Delta' \log \Delta'$ and $\Delta' := \log n / \log \log n$, we get that $\Bar{\Delta} = \Theta(\Delta')$ and that $\Bar{\Delta} = \Theta(\log n / \log \log n)$. Then, the lower bound results in $\Omega(\min\{\log n / \log \log n,  \log_{\log n} n  \}) = \Omega(\log n / \log \log n)$, and $G$ has chromatic number $\Theta(\Delta / \log \Delta) = \Theta(\log n / \log \log n)$.

\section{Preliminaries}\label{sec:preliminaries}
\paragraph{The LOCAL model.}
In the LOCAL model of distributed computation, the nodes of an $n$-node (hyper)graph are provided with a unique ID, typically in $1,\ldots,n^c$ for some integer $c>0$. The typical assumptions on the initial knowledge is that each node is aware of its own ID, its own degree in the graph (i.e., the number of neighboring nodes), the maximum degree $\Delta$, and the total number of nodes in the (hyper)graph. (In the case of hypergraphs, nodes know also the rank $r$.) 
In the randomized version of the LOCAL model, nodes have access to an infinite string of random bits. In this setting, the computation proceeds in synchronous rounds, where at each round nodes exchange messages with neighbors and perform some local computation. The size of the messages and the local computation can be arbitrarily large. The runtime of a distributed algorithm running in the LOCAL model is determined by the time it is needed for the very last node to terminate and produce its own output. In the randomized setting, we require that the randomized distributed algorithm succeeds with high probability, that is with probability at least $1-1/n^c$ for any fixed constant $c \geq 1$.

\paragraph{The Supported LOCAL model.}
In the Supported LOCAL model, we are given a (hyper)graph $G = (V,E)$ (called the \emph{support graph}) and a sub(hyper)graph $G' = (V',E')$ of $G$ (called the \emph{input graph}). The number of nodes of $G$ is denoted by $n$, the degree of $G$ by $\Delta$, and the degree of $G'$ by $\Delta'$. The rank in $G$ is denoted by $r$ and the rank in $G'$ is denoted by $r'$. Initially, nodes know $G$ and which of their incident (hyper)edges belong to $G'$ (if any). Also, nodes know $\Delta'$, $r'$ and the total number of nodes in $G'$. The goal is to solve some graph problem of interest in the subgraph $G'$. The computation proceeds in synchronous rounds, where at each round nodes exchange messages with their neighbors in $G$ and perform some local computation. Then, as in the LOCAL model: the size of the messages and the local computation can be arbitrarily large; the runtime of a distributed algorithm running in the Supported LOCAL model is determined by the time it is needed for the very last node to terminate and produce its own output; in the randomized setting, we require that the randomized distributed algorithm succeeds with probability at least $1-1/n^c$ for any fixed constant $c \geq 1$.

\paragraph{Problems in the black-white formalism.}
In order to apply round elimination, it is required to describe a problem in a formal language, called black-white formalism, that we now describe (the actual definition of this formalism is more general, and here we present a version that is simplified but that still fits our needs).
A problem $\Pi$ described in the black-white formalism is a tuple $(\Sigma,C_W,C_B)$ satisfying the following. 
\begin{itemize}
    \item $\Sigma$ is a finite set of labels.
    \item $C_W$ is a set of multisets of elements of $\Sigma$, where all multisets have the same size. 
    \item $C_B$ is a set of multisets of elements of $\Sigma$, where all multisets have the same size. 
\end{itemize}
Let $d_W$ be the size of the multisets in $C_W$, and $d_B$ be the size of the multisets in $C_B$. We will use problems defined in the black-white formalism in two settings, and for distinguishing the two cases we will use different notations.
One case will be the one of bipartite $2$-colored graphs.
By \emph{bipartitely solving} a problem $\Pi$ on $G$ we mean the following.
\begin{itemize}
    \item $G = (W \cup B,E)$ is a graph that is properly $2$-colored, and in particular each node $v \in W$ is labeled $c(v) = W$, and each node $v \in B$ is labeled $c(v) = B$.
    \item The task is to assign a label $o(e) \in \Sigma$ to each edge $e \in E$ such that, for each node $v \in W$ that has degree exactly $d_W$ (resp.\ $v \in B$ that has degree exactly $d_B$) it holds that the multiset of incident labels is in $C_W$ (resp.\ in $C_B$).
\end{itemize}
An assignment of labels to edges satisfying these requirements is called \emph{bipartite solution}.
The second case will be the one of hypergraphs. By \emph{non-bipartitely solving} a problem $\Pi$ we mean the following.
\begin{itemize}
    \item $G = (V,E)$ is a hypergraph.
    \item The task is to assign a label $o(v,e) \in \Sigma$ to each node-hyperedge pair $(v,e) \in V \times E$, such that, for each node $v \in V$ that has degree exactly $d_W$ (resp.\ for each hyperedge $e \in E$ that has rank exactly $d_B$) it holds that the multiset of incident labels is in $C_W$ (resp.\ in $C_B$).
\end{itemize}
An assignment of labels to node-edge pairs satisfying these requirements is called \emph{non-bipartite solution}.
In other words, non-bipartitely solving a problem $\Pi$ on a hypergraph $G$ means to bipartitely solve $\Pi$ on the incidence graph of $G$, and a non-bipartite solution for a problem $\Pi$ on a hypergraph $G$ is a bipartite solution for $\Pi$ on the incidence graph of $G$. 
We will use the term \emph{white constraint} in order to refer to $C_W$, and \emph{black constraint} in order to refer to $C_B$. In the case of graphs, we may use the term \emph{node constraint} to refer to $C_W$ and \emph{edge constraint} to refer to $C_B$.

Note that, in the variant of black-white formalism that we provided, all white nodes with degree different from $d_W$ do not need to satisfy any constraint, and all black nodes with degree different from $d_B$ do not need to satisfy any constraint. Moreover, nodes do not receive additional labels as inputs.

We will use the term (black or white) \emph{configuration} in order to refer to a multiset contained in a (black or white) constraint. We may represent configurations by using multisets, or by using regular expressions. For example $\{\A,\B,\C,\D\}$ is a configuration, and $\A \s \B \s \C \s \D$ is the same configuration. We may use regular expressions of the form, e.g., $[\A \B] \s [\C \D] \s \E$, to denote all configurations in $\{ \A \s \C\s \E, \A \s \D \s \E, \B \s \C \s \E, \B \s \D \s \E\}$, and we call such configurations \emph{condensed} configurations. 
For two labels $\X,\Y$, we say that $\X$ is \emph{at least as strong as} $\Y$ w.r.t.\ a constraint $C$ if, for all configurations containing $\Y$, it holds that by replacing an arbitrary amount of $\Y$ with $\X$ we obtain a configuration that is also in $C$.
The \emph{diagram} of $\Pi$ w.r.t.\ a constraint $C$ is a directed graph representing the strength relation: there is a directed path from $\Y$ to $\X$ if $\X$ is at least as strong as $\Y$. We call a set $S$ of labels \emph{right-closed} w.r.t.\ a diagram if, the fact that a label $\ell$ is in $S$, implies that all nodes reachable from $\ell$ in the diagram are also in $S$.
An example of encoding of a problem in this formalism, and an example of diagram, is presented in \Cref{sec:example-black-white}.

Let $\Pi' := (\Sigma', C'_W,C'_B)$ and $\Pi := (\Sigma,C_W,C_B)$ be two problems in the black-white formalism, both having white configurations of size $d_W$ and black configurations of size $d_B$.
Let $\Vec{C_W} = \{ (\ell_1,\ldots,\ell_{d_W}) \mid \{\ell_1,\ldots,\ell_{d_W}\} \in C_W  \}$,
that is we treat the multisets in $C_W$ as ordered tuples. Let $\Vec{C'_W}$ be defined analogously. The problem $\Pi'$ is a \emph{relaxation} of $\Pi$ if there exists a function $f : \Vec{C_W} \rightarrow \Vec{C'_W}$ satisfying the following.
For each label $\ell \in \Sigma$, let $r(\ell) := \{ \ell'_i \in \Sigma' \mid \exists C = (\ell_1,\ldots,\ell_{d_W}) \in \Vec{C_W}, C' = (\ell'_1,\ldots,\ell'_{d_W}) \in \Vec{C'_W}, f(C) = C', \exists i \text{ s.t. } \ell_i = \ell  \}$, that is, the set $r(\ell)$ contains all possible labels to which $\ell$ is mapped to, when using the mapping $f$.
For every black configuration $\{\ell_1,\ldots,\ell_{d_B}\}$, it is required that any choice over $r(\ell_1) \times \ldots \times r(\ell_{d_B})$ is in $C_B$. In other words, $\Pi'$ is a relaxation of $\Pi$ if there exists a way to map the configurations of white nodes in such a way that, if a solution is valid for $\Pi$, then the obtained solution is valid for $\Pi'$.

\paragraph{Black and white algorithms.}
In the context of algorithms for bipartite $2$-colored graphs, a white (resp.\ black) algorithm with runtime $T$ for a problem $\Pi$ (in the black-white formalism) on a graph $G$ is a function that takes as input the radius-$T$ neighborhood of a white (resp.\ black) node $v$ and provides a labeling for the edges incident to $v$, such that the constraints of $\Pi$ are satisfied on all the white and all the black nodes.

In other words, a white algorithm running on a bipartite graph is an algorithm where (black and white) nodes communicate for $T$ rounds, and then the nodes responsible for assigning an output for the edges of the graph are solely the white nodes, and the black nodes do not even need to know the outputs for their incident edges.

\paragraph{Round elimination.}
For the purpose of understanding the main content of our paper, it is sufficient to know that there exists a function $\mathrm{RE}$ that receives as input a problem $\Pi$ in the black-white formalism, and outputs a problem $\Pi' := \mathrm{RE}(\Pi)$ in a mechanical and clearly defined way, such that if the white (resp.\ black) configurations of $\Pi$ have size $\Delta$ (resp.\ $r$), then also the white (resp.\ black) configurations of $\Pi'$ have size $\Delta$ (resp.\ $r$). The actual definition of this function, and how the complexity of $\Pi'$ is related to $\Pi$, is only required when diving into the proofs of \Cref{sec:re}. Hence, we defer the definition of such a function to that section.

\paragraph{Lower bound sequence.}
Let $\Pi_0,\ldots,\Pi_k$ be a sequence of problems in the black-white formalism. This sequence is a \emph{lower bound sequence} if, for each $1 \le i \le k$, the problem $\Pi_{i}$ is a relaxation of $\mathrm{RE}(\Pi_{i-1})$. In this paper, we will not explicitly compute lower bound sequences. In fact, for the purposes of our proofs, we will only need to prove statements about the last problem of some sequences that have already been defined in different papers.

\paragraph{High-girth low-independence graphs.}
Throughout the paper, we will exploit the existence of a specific family of graphs. We report a known result from graph theory, shown in \cite{Alon10}.
\begin{lemma}[\cite{Alon10}]\label{lem:graph-family}
    There exists two constants $\alpha$ and $\epsilon$ such that, for any $n$ and $\Delta$ satisfying that $n \Delta$ is even and that $2 \le \Delta < n$, there exists a $\Delta$-regular graph of $n$ nodes that has girth at least $\varepsilon \log_\Delta n$ and independence number at most $\alpha \frac{n \log \Delta }{\Delta}$.
\end{lemma}

\section{A New Technique}\label{sec:new-technique}
In this section, we prove that understanding whether a problem $\Pi$ can be solved in $0$ rounds in the Supported LOCAL model is equivalent to understanding whether another problem $\bar{\Pi}$ admits a solution in the support graph $G$. The latter is conceptually easier than the former, since it does not require to think about distributed algorithms that run on subgraphs, but it just requires to think about existence of solutions. We note that a similar approach, but carefully crafted for the case of the sinkless orientation problem (which behaves significantly nicer in the round elimination framework), has been used in \cite{BalliuKKLOPPR0S23}.

\begin{definition}\label{def:lift}
    Let $\Pi$ be a problem in the black-white formalism, where the white configurations have size $\Delta'$ and the black configurations have size $r'$. For a pair of integers $\Delta \ge \Delta'$ and $r \ge r'$, the problem $\bar{\Pi} := \mathrm{lift}_{\Delta,r}(\Pi)$ is a problem in the black-white formalism defined as follows.
    \begin{itemize}
        \item $\Sigma_{\bar{\Pi}} = \{ L \mid L \subseteq \Sigma_{\Pi} \land L \neq \emptyset \land L \text{ is right-closed w.r.t. the black diagram of } \Pi \}$, that is, the set of labels of $\bar{\Pi}$ contains all possible non-empty subsets of the labels of $\Pi$ that are right-closed w.r.t.\ the black diagram of $\Pi$. These labels are called \emph{label-sets}.
        \item The black constraint of $\bar{\Pi}$ contains all multisets $C = \{L_1,\ldots,L_r\}$ of size $r$ satisfying the following. For all subsets $S = \{L_{i_1},\ldots,L_{i_{r'}}\}$ of $C$ of size $r'$, for any choice $\ell_1 \in L_{i_1},\ldots, \ell_{r'} \in L_{i_{r'}}$, it must hold that the configuration $\{\ell_1,\ldots,\ell_{r'}\}$ is in the black constraint of $\Pi$.
        \item The white constraint of $\bar{\Pi}$ contains all multisets $C = \{L_1,\ldots,L_\Delta\}$ of size $\Delta$ satisfying the following. For all subsets $S = \{L_{i_1},\ldots,L_{i_{\Delta'}}\}$ of $C$ of size $\Delta'$, there exists a choice $\ell_1 \in L_{i_1},\ldots, \ell_{\Delta'} \in L_{i_{\Delta'}}$ such that the configuration $\{\ell_1,\ldots,\ell_{\Delta'}\}$ is in the white constraint of $\Pi$.
    \end{itemize}
\end{definition}
\begin{restatable}{theorem}{lift}\label{th:lift-bipartite}
    Let $r$, $\Delta$, and $n$ be integers. Let $G$ be a $(\Delta,r)$-biregular graph of $n$ nodes, and let $\Pi$ be a problem in the black-white formalism satisfying that the white configurations have size $\Delta' \le \Delta$ and the black configurations have size $r' \le r$. The problem $\Pi$ can be bipartitely solved in $0$ rounds by a white algorithm in the Supported LOCAL model on $G$ if and only if there exists a bipartite solution for $\bar{\Pi} := \mathrm{lift}_{\Delta,r}(\Pi)$ on $G$.
\end{restatable}
\begin{proof}    
    We first prove that the existence of a bipartite solution for $\bar{\Pi}$ on $G$ implies a $0$-round white algorithm for $\Pi$. We define the algorithm $\mathcal{A}$ on the white nodes $v$ on $G = (V,E)$ as follows. Node $v$ computes a solution $S$ for $\bar{\Pi}$ without communication. This is possible since $v$ knows $G$. Note that this operation does not depend on the input of $v$, but solely on $G$, and hence we get that all nodes compute the same solution $S$. Then, node $v$ considers its incident edges that are part of the input graph $G'$. 
    If the count of these edges is not $\Delta'$, then $v$, for each edge $e$, outputs an arbitrary element from the set $L_e$ assigned to $e$ in the solution for $\Bar{\Pi}$.
    Otherwise, let the edges be $e_1,\ldots,e_{\Delta'}$. Let $L_1,\ldots,L_{\Delta'}$ be the sets of labels assigned to these edges in the solution $S$.
    By the definition of the white constraint of $\bar{\Pi}$, there exists a choice of labels $\ell_1 \in L_{1},\ldots, \ell_{\Delta'} \in L_{{\Delta'}}$ satisfying that the configuration $\{\ell_1,\ldots,\ell_{\Delta'}\}$ is in the white constraint of $\Pi$. Node $v$ outputs $\ell_1$ on $e_1$, $\ell_2$ on $e_2$, and so on. By construction, the white constraint of $\Pi$ is satisfied on all nodes. Moreover, by the definition of the black constraint of $\bar{\Pi}$, any output given by the white nodes satisfies the constraints of $\Pi$ on the black nodes.

    We now prove that, given a white algorithm $\mathcal{A}$ that solves $\Pi$ in $0$ rounds, we can find a bipartite solution for $\bar{\Pi}$ on $G$. While for proving the existence of a solution for $\bar{\Pi}$ it would be sufficient to provide a centralized algorithm that computes it, in the following we provide a distributed algorithm that finds a solution for $\bar{\Pi}$. Each white node $v$ computes a solution as follows. For each edge $e$ incident to $v$ on $G$, node $v$ initializes a set $L_e$ to be the empty set.  Then, node $v$ enumerates all possible choices of $\Delta'$ edges $e_1,\ldots,e_{\Delta'}$ over its $\Delta$ edges. Observe that, since $\mathcal{A}$ runs in $0$ rounds, its output on $v$ solely depends on $G$, $n$, $\Delta$, $r$, $\Delta'$, $r'$, and on which edges of $v$ are selected to be in $G'$. Thus, $v$ picks an arbitrary subgraph $G'$ of $G$ that has white degree bounded by $\Delta$', black degree bounded by $r'$, and that includes the edges $e_1,\ldots,e_{\Delta'}$, and computes (without communication) the output that $\mathcal{A}$ would provide on $G'$. Let this output be $\ell_1,\ldots,\ell_{\Delta'}$. Node $v$ adds $\ell_1$ to $L_{e_1}$, $\ell_2$ to $L_{e_2}$, and so on. At the end, for each edge $e$ of $G$, node $v$ outputs $L_e$.
    
    While the obtained sets are not necessarily right-closed, we first prove that they satisfy the second and third property of \Cref{def:lift}. We will then show that we can add elements to the sets assigned to the edges in order to make them right-closed, while still satisfying the second and third property.
    
    We start by proving that the white constraint of $\bar{\Pi}$ is satisfied. Suppose, for a contradiction, that the multiset $C = \{L_1,\ldots,L_\Delta\}$ obtained by $v$ does not satisfy the white constraint of $\bar{\Pi}$. This means that there exists a subset 
    $S = \{L_{i_1},\ldots,L_{i_{\Delta'}}\}$ of $C$ of size $\Delta'$ where all choices $\ell_1 \in L_{i_1},\ldots, \ell_{\Delta'} \in L_{i_{\Delta'}}$ satisfy that the configuration $\{\ell_1,\ldots,\ell_{\Delta'}\}$ is not in the white constraint of $\Pi$. This is in contradiction with the correctness of $\mathcal{A}$, since, when $v$ considered the edges $e_{i_1},\ldots,e_{i_{\Delta'}}$, it added to the sets $L_{i_1},\ldots,L_{i_{\Delta'}}$ the outputs $\ell_1,\ldots,\ell_{\Delta'}$ of $\mathcal{A}$.
    
    We now prove that the black constraint of $\bar{\Pi}$ is satisfied. 
    Suppose, for a contradiction, that the multiset $C = \{L_1,\ldots,L_r\}$ obtained by a black node $u$ does not satisfy the black constraint of $\bar{\Pi}$. This means that there exists a subset 
    $S = \{L_{i_1},\ldots,L_{i_{r'}}\}$ of $C$ of size $r'$ where there exists a choice $\ell_1 \in L_{i_1},\ldots, \ell_{r'} \in L_{i_{r'}}$ satisfying that the configuration $\{\ell_1,\ldots,\ell_{r'}\}$ is not in the black constraint of $\Pi$. We show that this implies that $\mathcal{A}$ must fail on some input graph $G'$.
    We construct the graph $G'$ as follows. For each edge $e_{i_j}$ of the black node $u$, let $v_j$ be the white node connected to $e_{i_j}$. We select $\Delta'$ edges of each $v_j$ in such a way that:
    \begin{itemize}
        \item The edge $e_{i_j}$ is selected;
        \item The other $\Delta'-1$ edges are selected in such a way that $\mathcal{A}$, when run on $v_j$, outputs $\ell_j$ on the edge $e_{i_j}$. 
    \end{itemize}
    By construction, such a choice exists. We complete $G'$ in an arbitrary way, such that the maximum white degree is bounded by $\Delta'$ and the black degree is bounded by $r'$. Observe that $\mathcal{A}$ must fail on $G'$, and in particular on node $u$, since it gets the configuration $\{\ell_1,\ldots,\ell_{r'}\}$.

    We now prove that we can add elements to the sets assigned to the edges in order to make them right-closed w.r.t.\ the black diagram of $\Pi$, while still preserving the second and third requirement of \Cref{def:lift}. Let $L$ be the set assigned to edge $e$. For each element $\ell \in L$, we add to $L$ all the successors of $\ell$ in the black diagram of $\Pi$. Clearly, since we did not remove elements from $L$, the white constraint of $\Bar{\Pi}$ is still satisfied. Then, by the definition of the black diagram of $\Pi$, it holds that, if a configuration $C = \{\ell_1,\ldots,\ell_{r'}\}$ is allowed, then any configuration obtained by replacing arbitrary elements of $C$ with elements that can be reached by them in the diagram is also allowed. Thus, the black constraint of $\Bar{\Pi}$ is still satisfied.
\end{proof}
Note that, by following the exact same arguments, we obtain the following corollary.
\begin{corollary}\label{th:lift}
    Let $r$, $\Delta$, and $n$ be integers. Let $G$ be a $\Delta$-regular $r$-uniform hypergraph with $n$ nodes. Let $\Pi$ be a problem in the black-white formalism satisfying that the white configurations have size $\Delta' \le \Delta$ and the black configurations have size $r' \le r$. The problem $\Pi$ can be non-bipartitely solved in $0$ rounds by a white algorithm in the Supported LOCAL model on $G$ if and only if there exists a non-bipartite solution for $\bar{\Pi} := \mathrm{lift}_{\Delta,r}(\Pi)$ on $G$.
\end{corollary}

In \Cref{sec:re}, we will prove that a lower bound sequence $\Pi_0,\ldots, \Pi_k$, combined with the non-$0$-round solvability of $\Pi_k$ in the Supported LOCAL model, implies a lower bound of $\min\{2k,\frac{g-4}{2}\}$ deterministic rounds in support graphs of girth $g$.
As we will see, the sequence $\Pi_0,\ldots, \Pi_k$ is defined in the same exact way as when using round elimination in the standard LOCAL model. Hence, we can reuse lower bound sequences that are already known from LOCAL lower bounds, and we only need to argue about non-$0$-round solvability of $\Pi_k$.
In the following, we assume that, on instances of size $n$, the ID space is exactly $\{1,\ldots,n\}$. This is without loss of generality: since all nodes know the support graph, given an ID assignment over a larger domain, the nodes can compute a consistent ID assigment over $\{1,\ldots,n\}$ without communication. In \Cref{sec:rand}, we will prove that, if a problem $\Pi$ has deterministic complexity $D_\Pi(n)$, and randomized complexity $R_\Pi(n)$, then $D_\Pi(n) \le R_\Pi(2^{3n^2})$. 
We now combine these results with \Cref{th:lift-bipartite} to obtain the following.
\begin{restatable}{theorem}{approachBipartite}\label{thm:approach-bipartite}
Let $\Pi$ be a problem in the black-white formalism. Let $\Delta'$ be the size of the multisets in the white constraint of $\Pi$, and let $r'$ be the size of the multisets in the black constraint of $\Pi$. Assume $\Delta' \ge 2$ and $r' \ge 2$.
Let $\Delta$ and $r$ be arbitrary integers satisfying $\Delta \ge \Delta'$ and $r \ge r'$.
Let $\Pi_0,\,\ldots,\Pi_k$ be a lower bound sequence, where $\Pi_0 = \Pi$. Let $\Pi'$ be some relaxation of $\Pi_k$.
Assume that there exist constants $c > 0$ and $\varepsilon > 0$, and a family $\mathcal{G}$ of graphs satisfying that, for any $n \ge (r\Delta)^c$, there exists a graph $G \in \mathcal{G}$ satisfying the following.
\begin{itemize}
    \item $G$ has at least $n / (\Delta r)^c$ and at most $n$ nodes.
    \item $G$ is $(\Delta,r)$-biregular.
    \item $G$ has girth at least $\varepsilon \log_{\Delta r} n$.
    \item There does not exist a valid bipartite solution for $\mathrm{lift}_{\Delta,r}(\Pi')$ on $G$.
\end{itemize}
    Then, for all $n$, there exists a bipartite graph in which white nodes have degree bounded by $\Delta$, black nodes have degree bounded by $r$, and there are exactly $n$ nodes, where any algorithm for bipartite-solving $\Pi$ requires at least $\min\{2k,(\varepsilon(\log_{\Delta r} (n) - c) - 4) / 2\} - 1$ deterministic rounds and $\min\{2k,(\varepsilon(\log_{\Delta r} (\sqrt{\frac{1}{3}\log n}) - c) - 4) / 2\} - 1$ randomized rounds in the Supported LOCAL model.
\end{restatable}
\begin{proof}
We start by proving the deterministic bound. If $n < (r \Delta )^c$, then the claimed lower bound is at most $0$, and hence it trivially holds. Thus, in the following, we assume $n \ge (r \Delta)^c$. We pick a graph $G$ satisfying the conditions of the statement, and let $n'$ be the number of nodes in $G$.
    We consider the graph $\bar{G}$ containing two connected components: one is $G$, and the other is an arbitrary tree where white nodes have maximum degree $\Delta$ and black nodes have maximum degree $r$, containing $n - n'$ nodes. We obtain a graph with exactly $n$ nodes. By assumption, there is no solution for $\mathrm{lift}_{\Delta,r}(\Pi')$ on $G$ (and hence neither on $\bar{G}$), and by \Cref{th:lift-bipartite} this implies that $\Pi'$ cannot be solved in $0$ rounds by a white algorithm on $G$.
    
    Hence, by \Cref{lem:re-works}, solving $\Pi$ on the component $G$ requires at least $\min\{2k,\frac{g-4}{2}\}$ rounds for a white algorithm, where $g$ is the girth of $G$. By assumption, $g$ is at least $\varepsilon \log_{\Delta r} (n / (\Delta r)^c) = \varepsilon ( \log_{\Delta r} (n) - c)$. Hence, solving $\Pi$ requires at least $\min\{2k,(\varepsilon(\log_{\Delta r} (n) - c) - 4) / 2\}$ deterministic rounds for a white algorithm on the component $G$, and hence also on $\Bar{G}$. Since a black algorithm can be converted, by spending $1$ additional round, into a white algorithm, we obtain the claimed deterministic lower bound.

    We now lower bound the randomized complexity.
    By \Cref{lem:derand-supported}, $D_\Pi(n) \le R_\Pi(2^{3n^2})$, and hence $R_\Pi(n) \ge D_\Pi(\sqrt{\frac{1}{3}\log n})$.
    Hence, the claim follows.
\end{proof}

By replacing, in the proof of \Cref{thm:approach-bipartite}, the application of \Cref{th:lift-bipartite}, \Cref{lem:re-works}, and \Cref{lem:derand-supported}, with \Cref{th:lift},  \Cref{cor:re-works-hyper}, and \Cref{lem:hyperderandomization}, we obtain the following.
\begin{restatable}{corollary}{approachHypergraphs}\label{thm:approach-hypergraphs}
Let $\Pi$ be a problem in the black-white formalism. Let $\Delta'$ be the size of the multisets in the white constraint of $\Pi$, and let $r'$ be the size of the multisets in the black constraint of $\Pi$. Assume $\Delta' \ge 2$ and $r' \ge 2$.
Let $\Delta$ and $r$ be arbitrary integers satisfying $\Delta \ge \Delta'$ and $r \ge r'$.
Let $\Pi_0,\,\ldots,\Pi_k$ be a lower bound sequence, where $\Pi_0 = \Pi$. Let $\Pi'$ be some relaxation of $\Pi_k$.
Assume that there exist constants $c > 0$ and $\varepsilon > 0$, and a family $\mathcal{G}$ of hypergraphs satisfying that, for any $n \ge \Delta^c$, there exists a hypergraph $G \in \mathcal{G}$ satisfying the following: $G$ has at least $n / (\Delta r)^c$ and at most $n$ nodes; $G$ is a $\Delta$-regular $r$-uniform linear hypergraph; $G$ has girth at least $\varepsilon \log_{\Delta r} n$; There does not exist a valid non-bipartite solution for $\mathrm{lift}_{\Delta,r}(\Pi')$ on $G$.

Then, for all $n$, there exists a hypergraph of maximum degree $\Delta$ and rank $r$, with exactly $n$ nodes, where any algorithm for non-bipartite-solving $\Pi$ requires at least $\min\{k,(\varepsilon(\log_{\Delta r} (n) - c) - 4) / 2\} - 1$ deterministic rounds and $\min\{k,(\varepsilon(\log_{\Delta r} (\sqrt[3]{\frac{1}{4}\log n}) - c) - 4) / 2\} - 1$ randomized rounds in the Supported LOCAL model.
\end{restatable}

\section{Variants of Maximal Matching}\label{sec:matching}
In this section, we prove lower bounds for $x$-maximal $y$-matchings in the Supported LOCAL model.
More in detail, we prove the following theorem.
\begin{restatable}{theorem}{lbXyMatching}\label{thm:lb-xy-matching}
    Let $\mathcal{G}$ be the family of bipartite $2$-colored $\Delta$-regular graphs of girth $\Omega(\log_\Delta n)$.
    Assume that the support graph is from $\mathcal{G}$, and that the input graph has degree $\Delta'$ satisfying $\Delta \ge c\Delta'$ for some large-enough constant $c$.
    Then, the $x$-maximal $y$-matching problem requires $\Omega(\min\{(\Delta'-x) / y,\log_\Delta n\})$ rounds in the deterministic Supported LOCAL model and $\Omega(\min\{(\Delta'-x) / y,\log_\Delta \log n\})$ rounds in the randomized Supported LOCAL model.
\end{restatable}

\subsection {Problem Definition}
In order to prove \Cref{thm:lb-xy-matching}, we consider a family of problems that we will later show to be strongly related with  $x$-maximal $y$-matchings.
\paragraph{A family of problems in the black-white formalism.} We define $\Pi_\Delta(x,y)$ as follows.
\begin{definition}[The problem $\Pi_\Delta(x,y)$]
The problem $\Pi_\Delta(x,y)$ is defined via the following white and black constraints.
\begin{equation*}
	\begin{aligned}
		\begin{aligned}
			\Pi^W_\Delta(x,y)\text{:}\\
	&\X^{y-1} \s \M \s \O^{\Delta-y} \\
	&\X^y \s \O^{x} \s \P^{\Delta-y-x} \\
	&\X^y \s \Z \s \O^{\Delta-y-1}
 		\end{aligned}
		\qquad
		\begin{aligned}
  			\Pi^B_\Delta(x,y)\text{:}\\
	&[\M \Z \P \O \X]^{y-1} \s [\M \X] \s [\P \O \X]^{\Delta-y} \\
	&[\M \Z \P \O \X]^y \s [\P \O \X]^{x} \s [\O \X]^{\Delta-y-x} \\
	&[\M \Z \P \O \X]^y \s [\X] \s [\P \O \X]^{\Delta-y-1}
		\end{aligned}
	\end{aligned}
\end{equation*}
\end{definition}
We observe that, by increasing the value of the parameters $x$ and $y$, we obtain a relaxed problem.
\begin{observation}\label{obs:notharder}
 For any $x' \ge x$ and any $y' \ge y$, the problem $\Pi_\Delta(x',y')$ is a relaxation of $\Pi_\Delta(x,y)$.
\end{observation}
\begin{proof}
We show that a white node $v$ can convert a solution for $\Pi_\Delta(x,y)$ into a solution for $\Pi_\Delta(x',y')$ without communication.
    Each white node $v$ operates as follows. If its output is $\X^{y-1} \s \M \s \O^{\Delta-y}$, it converts the required amount of $\O$ into $\X$ in an arbitrary way, in order to obtain the configuration $\X^{y'-1} \s \M \s \O^{\Delta-y'}$. If its output is $\X^y \s \O^{x} \s \P^{\Delta-y-x}$, it converts the required amount of $\P$ into $\O$ or $\X$, and the required amount of $\O$ into $\X$, in an arbitrary way, in order to obtain the configuration $\X^{y'} \s \O^{x'} \s \P^{\Delta-y'-x'}$.  If its output is $\X^y \s \Z \s \O^{\Delta-y-1}$, it converts the required amount of $\O$ into $\X$, in an arbitrary way, in order to obtain the configuration $\X^{y'} \s \Z \s \O^{\Delta-{y'}-1}$.
    Observe that, since we only replace labels with ones that can be reached from them in the black diagram of the problem (depicted in \Cref{fig:BD_xy_matching}), then the constraints are still satisfied on the black nodes.
\end{proof}

\paragraph{What is known about these problems. } In \cite{trulytight}, it has been shown that the problem $\Pi_\Delta(x,y)$ is strictly related to $x$-maximal $y$-matching, and in particular that, given a solution for the latter, it is possible to solve the former in just $2$ additional rounds of communication.
\begin{lemma}[Lemma 4.2 and 4.3 of \cite{trulytight}]\label{lem:xy-to-re}
    In the LOCAL model, given a solution for $x$-maximal $y$-matching, it is possible to solve $\Pi_\Delta(x,y)$ in $2$ rounds.
\end{lemma}

In \cite{trulytight}, it is shown how problems $\Pi_\Delta(x,y)$ with different parameters are related.
\begin{lemma}[Lemma 4.1 of \cite{trulytight}, rephrased]\label{lem:trulytight}
Assume $x+2y \le \Delta$.
    Then, for any $\Delta \ge 2$, $1 \le y \le \Delta - 1$, and $0 \le x \le \Delta - y$, $\Pi_\Delta(x+y,y)$ is a relaxation of $\mathrm{RE}(\Pi_\Delta(x,y))$.
\end{lemma}

By applying \Cref{lem:trulytight} for $k$ times, we obtain the following.
\begin{corollary}\label{cor:xy-re-sequence}
    Assume $x + (k+1)y \le \Delta$. Then, there exists a lower bound sequence $\Pi_1,\ldots,\Pi_k$, where $\Pi_1 = \Pi_\Delta(x,y)$ and $\Pi_k = \Pi_\Delta(x+ky,y)$.
\end{corollary}

\subsection{A Lower Bound for the Supported LOCAL Model}
In the remainder of the section we prove the claimed lower bound for $x$-maximal $y$-matchings in the Supported LOCAL model. More in detail, by exploiting \Cref{thm:approach-bipartite}, we will prove that, for any fixed $\Delta$, the problem $\Pi_{\Delta'}(x,y)$, in the case where the support graph has degree bounded by $\Delta$ and the input graph has degree bounded by $\Delta'$ satisfying $\Delta \ge c \Delta'$ for some large-enough constant $c$, requires at least $\min\{k, \varepsilon \log_\Delta n\} - 1$ deterministic and $\min\{k, \varepsilon \log_\Delta \log n\} - 1$ randomized  rounds, for some absolute constant $\varepsilon$ (i.e., that does not depend on $\Delta$ and $n$), where $k := \lfloor \frac{\Delta'-x}{y} \rfloor - 2$. 
By \Cref{lem:xy-to-re}, by decreasing the lower bound by $2$, we obtain a lower bound that holds also for $x$-maximal $y$-matchings.
Note that such a statement implies \Cref{thm:lb-xy-matching}. Hence, in the following, let $\Delta$ be a fixed value.

Observe that $x + ky \le x + (\frac{\Delta'-x}{y} - 2) y \le x + \Delta' - x - 2y = \Delta' - 2y \le \Delta' - 1 - y =: x'$. Hence, by \Cref{obs:notharder}, $\Pi_{\Delta'}(x',y)$ is a relaxation of $\Pi_k$. Thus, there exists a lower bound sequence of length $k := \lfloor \frac{\Delta'-x}{y} \rfloor - 2$, where the first problem is $\Pi_{\Delta'}(x,y)$ and the last problem is $\Pi_{\Delta'}(x',y)$. Hence, by \Cref{thm:approach-bipartite} (applied with $c=1$), if we prove that, for any $n \ge \Delta^2$, there exists a $\Delta$-regular bipartite graph $G$ with at least $n / \Delta^2$ nodes and at most $n$ nodes, of girth at least $\epsilon' \log_\Delta n$ (for some absolute constant $\varepsilon'$), and where $\mathrm{lift}_{\Delta,\Delta}(\Pi_{\Delta'}(x',y))$ has no bipartite solution, then we obtain a lower bound for $\Pi_{\Delta'}(x,y)$, as desired. 

In the remainder of the section, we prove that such graphs exist. For a given $n$, we consider a number $n' \in \{n,n-1,n-2,n-3\}$ such that $n' \Delta /2$ is even. We take a graph of size $n'/2$ from the family given by \Cref{lem:graph-family}, and then we take its bipartite double cover. We obtain a graph $G$ of size $n'$ that is $\Delta$-biregular and that satisfies the requirement on the girth. Note that $n/\Delta^2 \le n/4 \le n' \le n$, as required. Let $\Pi:= \Pi_{\Delta'}(x',y)$, and let $\Bar{\Pi} := \mathrm{lift}_{\Delta,\Delta}(\Pi)$. What remains to be done is to show that, on $G$, the problem $\Bar{\Pi}$ is unsolvable.
For ease of reading, we report here the problem $\Pi$.
\begin{equation*}
	\begin{aligned}
		\begin{aligned}
			\Pi^W_{\Delta'}(x',y)\text{:}\\
	&\X^{y-1} \s \M \s \O^{\Delta'-y} \\
	&\X^y \s \O^{\Delta' -1 -y} \s \P \\
	&\X^y \s \Z \s \O^{\Delta-y-1}
 		\end{aligned}
		\qquad
		\begin{aligned}
  			\Pi^B_{\Delta'}(x',y)\text{:}\\
	&[\M \Z \P \O \X]^{y-1} \s [\M \X] \s [\P \O \X]^{\Delta'-y} \\
	&[\M \Z \P \O \X]^y \s [\P \O \X]^{\Delta' - 1 - y} \s [\O \X] \\
	&[\M \Z \P \O \X]^y \s [\X] \s [\P \O \X]^{\Delta'-y-1}
		\end{aligned}
	\end{aligned}
\end{equation*}
In the following, we denote with $2n$ the number of nodes of $G$.
    Assume, for a contradiction, that there exists a bipartite solution for $\Bar{\Pi}$ on $G$. For each edge $e$, let $S_e$ be the set of labels (of $\Pi$) assigned to $e$.  In the following,  with $\mybox{L_1 \ldots L_k}$,  we denote the set $\{L_1,\ldots,L_k\}$.
    
    Recall that, by \Cref{th:lift}, each set $S_e$ is non-empty and right-closed w.r.t.\ the black diagram of $\Pi$, which is shown in \Cref{fig:BD_xy_matching}. That is, if in the diagram there is an arrow from a label $A$ to a label $B$, then $A \in S_e \implies B \in S_e$. Thus, each set $S_e$ can only be one of the following label-sets: $\bX$,$\bOX$, $\bMX$, $\bMOX$, $\bPOX$, $\bMPOX$, $\bZMPOX$.

    \begin{figure}[t]
        \centering
        \includegraphics[width = 0.3\textwidth]{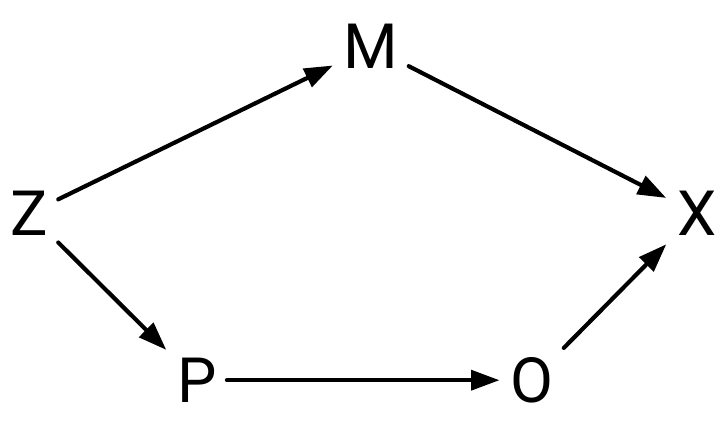}
        \caption{Black diagram of $\Pi$.}
        \label{fig:BD_xy_matching}
    \end{figure}

    In order to prove that there cannot be a solution for $\Bar{\Pi}$, we apply the following strategy. First, we prove that, in order to satisfy the constraints of $\Bar{\Pi}$ on the white nodes, \emph{at least} some amount of edges need to be labeled with label-sets that contain $\P$. Then, we prove that, in order to satisfy the constraints of $\Bar{\Pi}$ on the black nodes, \emph{at most} some amount of edges need to be labeled with label-sets that contain $\P$. We then show that, for $c$ large enough, the two provided bounds are not compatible, implying a contradiction.
    We first prove the following lemma.

\begin{lemma}\label{lem:ub-m}
    Let $2n$ be the number of nodes of $G$.
    Any solution for $\Bar{\Pi}$ must satisfy that at most  $n y$ edges are labeled with label-sets containing $\M$.
\end{lemma}
\begin{proof}
    Let $v$ be a black node. We prove that at most $y$ edges incident to $v$ are labeled with label-sets containing $\M$. Suppose for a contradiction that at least $y+1$ edges incident to $v$ have label-sets containing $\M$. Then, since each configuration in $\Pi^B_{\Delta'}(x',y)$ allows at most $y$ $\M$ labels, we obtain that 
    there exists a choice of $\Delta'$ edges incident to $v$ satisfying the following. The $\Delta'$ edges have label-sets assigned $L_1,\ldots,L_{\Delta'}$, and there exists a choice $\ell_1 \in L_1,\ldots,\ell_{\Delta'} \in L_{\Delta'}$ that is not in $\Pi^B_{\Delta'}(x',y)$, which a contradiction with the definition of the black constraint of $\Bar{\Pi}$ in \Cref{def:lift}.  
\end{proof}
    

\begin{lemma}\label{lem:xy-lb-p}
    Let $2n$ be the number of nodes of $G$.
    Any solution for $\Bar{\Pi}$ must satisfy that at least  $n (\frac{\Delta - \Delta'}{2} - y)$ edges are labeled with label-sets containing $\P$, that is, with $\bPOX$, $\bMPOX$, or $\bZMPOX$.
\end{lemma}
\begin{proof}
    Observe that, since $\Delta' - y - x' = \Delta' - y - (\Delta' - 1 - y) = 1$, each configuration in $\Pi^W_{\Delta'}(x',y)$ requires at least one label in $\{\M,\P,\Z\}$. 
    Thus, incident to $v$, there are at most $\Delta'-1$ edges that are either $\bX$ or $\bOX$, since otherwise, similarly as in the proof of \Cref{lem:ub-m}, we could pick $\Delta'$ edges that do not contain any label in $\{\M,\P,\Z\}$, reaching a contradiction. 
    Since all the label-sets that are not $\bX$ or $\bOX$ contain either an $\M$ or a $\P$, then we conclude that there are at least $\Delta - \Delta' +1$ edges with assigned label-sets that contain $\M$ or $\P$, that is, $\bMX$, $\bMOX$, $\bPOX$, $\bMPOX$, $\bZMPOX$.

    We split the white nodes into two parts. A white node is called $M$-node if it has at least $(\Delta-\Delta')/2$ edges with assigned label-sets that contain the label $\M$, and it is called $P$-node otherwise.

    Observe that there cannot be too many $M$-nodes. In fact, suppose for a contradiction that there are at least $\frac{2ny}{\Delta-\Delta'}+1$ $M$-nodes. We obtain that the total number of edges having $\M$ in their label-sets is at least $(\frac{2ny}{\Delta-\Delta'}+1) \frac{\Delta-\Delta'}{2} > ny$, a contradiction with \Cref{lem:ub-m}. 
    
    We thus get that the amount of $P$-nodes is at least $n(1-\frac{2y}{\Delta-\Delta'})$. Since these nodes have at least $\Delta - \Delta' +1$ edges with assigned label-sets containing $\M$ or $\P$, but strictly less than $(\Delta-\Delta')/2$ edges with assigned label-sets containing $\M$, we obtain that $P$-nodes have, in total, at least $n(1-\frac{2y}{\Delta'-\Delta})\frac{\Delta-\Delta'}{2} = n (\frac{\Delta - \Delta'}{2} - y)$ edges whose assigned label-sets contain $\P$.
\end{proof}

\begin{lemma}\label{lem:xy-ub-p}
    Let $2n$ be the number of nodes of $G$.
    Any solution for $\Bar{\Pi}$ must satisfy that at most 
    $n (\Delta'-1)$
    edges are labeled with label-sets containing $\P$, that is, with $\bPOX$, $\bMPOX$, or $\bZMPOX$.
\end{lemma}
\begin{proof}
    Let $v$ be a black node. 
    Since $\Delta' - y - x' = 1$, we get that  $\P^{\Delta'}$ is not in $\Pi^B_{\Delta'}(x',y)$.
    By the definition of $\Bar{\Pi}$, we get the following. Let $e_1,\ldots,e_{\Delta'}$ be an arbitrary choice of $\Delta'$ edges incident to $v$, and let $L_1,\ldots,L_{\Delta'}$ be the label-sets assigned to these edges. Any choice $\ell_1 \in L_1,\ldots,\ell_{\Delta'} \in L_{\Delta'}$ must satisfy that $\ell_1 \ldots \ell_{\Delta'}$ is a configuration allowed by  $\Pi^B_{\Delta'}(x',y)$. We thus get that, incident to $v$, there can be at most $\Delta'-1$ edges with assigned label-sets containing $\P$, showing the lemma.
\end{proof}

We now prove that \Cref{lem:xy-lb-p} is in contradiction with \Cref{lem:xy-ub-p}.  
We fix $c = 5$, that is,  $\Delta = 5\Delta'$. From \Cref{lem:xy-lb-p}, we obtain that the amount $n_\P$ of edges with label-sets containing $\P$ can be bounded as follows.
\[
n_\P \ge n(\frac{\Delta-\Delta'}{2} -y )\ge n (2 \Delta'-y)\ge n\Delta' 
\]
From \Cref{lem:xy-ub-p}, we obtain the following bound.
\[
    n_\P \le n (\Delta'-1)
\]
Thus, we reach a contradiction.

\section{Arbdefective Coloring}\label{sec:arb}
In this section, we prove the following.
\begin{restatable}{theorem}{arbColoring}\label{thm:arb-coloring}
    Let $\mathcal{G}$ be the family of $\Delta$-regular graphs of girth $\Omega(\log_\Delta n)$.
    Assume that the support graph is from $\mathcal{G}$, and that the input graph has degree $\Delta'$. Let $k = \min\{\Delta', \eps \Delta / \log \Delta \}$, for some small-enough constant $\eps$.
    If $(\alpha+1)c \le k$, then the $\alpha$-arbdefective $c$-coloring problem, in the Supported LOCAL model, requires $\Omega(\log_\Delta n)$ deterministic rounds and $\Omega(\log_\Delta \log n)$ randomized rounds.
\end{restatable}

\subsection{Problem Definition}
In order to prove \Cref{thm:arb-coloring}, we consider a family of problems that we will later show to be strongly related with  $\alpha$-arbdefective $c$-coloring.
\paragraph{A family of problems in the black-white formalism.} 
We define a problem $\Pi_{\Delta}(c)$ in the black-white formalism as follows.
\begin{definition}[The problem $\Pi_\Delta(c)$]\label{def:arbdef}
Let $\mathcal{C} := \{1,\ldots,c\}$, and let $\Sigma_{\Pi_\Delta(c)} := \{ \X \} \cup \{ \ell(C) \mid C \subseteq \mathcal{C} \text{ and } C \neq \emptyset \}$.
The problem $\Pi_\Delta(c)$ is defined via the following white (node) and black (edge) constraints.
\begin{equation*}
	\begin{aligned}
		\begin{aligned}
			\Pi^W_\Delta(c)\text{:}\\
	&\ell(C)^{\Delta-x} \s \X^x,&& \text{where} \; x=|C|-1, \\ &&&\text{for all }  C \subseteq \mathcal{C}, C \neq \emptyset
 		\end{aligned}
   \qquad
		\begin{aligned}
  			\Pi^B_\Delta(c)\text{:}\\
	&\ell(C_1) \s \ell(C_2),&& \text{ for all $C_1, C_2$ s.t. } C_1 \cap C_2 = \emptyset, \\ &&& \text{ where } C_1, C_2 \neq \emptyset\\
    &\X \s \L,&& \text{ for all $\L \in \Sigma_{\Pi_\Delta(c)} $}
		\end{aligned}
	\end{aligned}
\end{equation*}
\end{definition}

\paragraph{What is known about these problems.}
In \cite{Balliu0KO22}, it has been shown that $\Pi_\Delta(c)$ is strictly related to $\alpha$-arbdefective $c$-coloring, and in particular that a solution for $\alpha$-arbdefective $c$-coloring can be converted in $0$ rounds into a solution for $\Pi_\Delta((\alpha+1)c)$, implying that $\Pi_\Delta((\alpha+1)c)$ is at least as easy as $\alpha$-arbdefective $c$-coloring.
\begin{lemma}[Theorem 8.2 (in the ArXiv version) of \cite{Balliu0KO22}]\label{lem:arbdef-to-bw}
    In the LOCAL model, given a solution for $\alpha$-arbdefective $c$-coloring, it is possible to solve $\Pi_\Delta((\alpha+1)c)$ in $0$ rounds.
\end{lemma}
Moreover, in \cite{Balliu0KO22}, it has been shown that, if $(\alpha+1)c \le \Delta$, then $\Pi_\Delta((\alpha+1)c)$ is a so-called \emph{fixed point} under round elimination, meaning that applying round elimination on the problem gives the problem itself, and hence that there exists a lower bound sequence of infinite length.
\begin{lemma}[Section 4 (in the ArXiv version) of \cite{Balliu0KO22}]\label{lem:arbdef-fp}
    Assume $(\alpha+1)c \le \Delta$. Then, $\mathrm{RE}(\Pi_\Delta((\alpha+1)c)) = \Pi_\Delta((\alpha+1)c)$.
\end{lemma}
\begin{corollary}\label{lem:arbdef-sequence}
    Assume $(\alpha+1)c \le \Delta$. Then, there exists a lower bound sequence of infinite length where all problems are equal to $\Pi_\Delta((\alpha+1)c)$.
\end{corollary}

\subsection{A Lower Bound for the Supported LOCAL Model}\label{apx:arbdef}
In order to prove \Cref{thm:arb-coloring}, we follow the same strategy as in the case of $x$-maximal $y$-matchings, that is, we operate as follows.
Let $k = \min\{\Delta', \eps \Delta / \log \Delta \}$, for some small-enough constant $\eps$.
We show that, for any $n$ and $\Delta$ such that $n \ge \Delta^2$, there exists a $\Delta$-regular graph $G$ with at least $n / \Delta^2$ nodes and at most $n$ nodes, of girth at least $\epsilon' \log_\Delta n$ (for some absolute constant $\varepsilon'$), and where,  assuming $(\alpha+1)c \le k$,  $\Pi' := \mathrm{lift}_{\Delta,2}(\Pi_{\Delta'}(\alpha+1)c))$ has no non-bipartite solution. For this purpose, we use exactly the graph family $\mathcal{G}$ given by \Cref{lem:graph-family}.
 
Hence, in the following, let $G \in \mathcal{G}$. We need to prove that, on $G$, $\Pi'$ is not non-bipartitely solvable.
We actually prove a different statement, that will be useful also in the next section. We will show that this statement implies what we want, that is, that $\Pi'$ admits no solution in $G$. We start by defining what we mean by solving a problem $\Pi$ on a subset of nodes $S$ of $G$.
\begin{definition}[$S$-solution of $\Pi$]\label{def:s-solution}
    Let $G$ be a graph, and let $\Pi$ be a problem defined on $G$ in the black-white formalism. A labeling of $G$ is an $S$-solution of $\Pi$ if the following holds.
    \begin{itemize}
        \item The node constraint of $\Pi$ is satisfied on all nodes of $S$.
        \item The edge constraint of $\Pi$ is satisfied on all edges that connect two nodes of $S$.
    \end{itemize}
\end{definition}
\begin{lemma}\label{lem:subgraph-coloring}
Let $S$ be a subset of nodes of $G$, and let $k \le \Delta'$ be an integer. Assume we are given an $S$-solution for  $\Pi' := \mathrm{lift}_{\Delta,2}(\Pi_{\Delta'}(k))$. 
    Then, it is possible to color the subgraph induced by the nodes in $S$ with $2k$ colors.
\end{lemma}
We start by showing that \Cref{lem:subgraph-coloring} implies that $\Pi'$ admits no solution in $G$. Later, we will prove \Cref{lem:subgraph-coloring}.
\begin{corollary}
   The problem $\Pi'$ admits no solution in $G$.
\end{corollary}
\begin{proof}
Let $V$ be the set of nodes of $G$, and consider $S = V$. By applying \Cref{lem:subgraph-coloring}, we obtain that we can color $G$ with $2(\alpha+1)c \le 2 \epsilon \Delta / \log \Delta $ colors, which, by choosing $\epsilon$ small enough, is less than the chromatic number of $G$, a contradiction.
\end{proof}
In the rest of the section, we prove \Cref{lem:subgraph-coloring}. We prove \Cref{lem:subgraph-coloring} in two steps. The first step is to show that, if we are given an $S$-solution for $\Pi'$, then we can convert it into an $S$-solution for $\Pi_{\Delta}((\alpha+1)c)$. 
The second step is to show that an $S$-solution for $\Pi_{\Delta}((\alpha+1)c)$ can be converted into a proper $(2(\alpha+1)c)$-coloring of the nodes in $S$. Note that these two statements imply \Cref{lem:subgraph-coloring}. In other words, the second step proves that $\Pi_{\Delta}(x)$ is strictly related to the problem of coloring a graph of degree $\Delta$ with $\Theta(x)$ colors. Based on this informal equivalence between $\Pi_{\Delta}(x)$ and $x$-coloring, we can informally restate the first step as proving that, if we can solve $x$-coloring in $0$ rounds in all subgraphs of maximum degree $\Delta'$, then we can solve $x$-coloring also on the support graph of degree~$\Delta$.

\begin{lemma}\label{lem:sol_on_d'}
Let $S$ be a subset of nodes of $G$, and let $k \le \Delta'$ be an integer. Assume there exists an $S$-solution for $\Pi' := \mathrm{lift}_{\Delta,2}(\Pi_{\Delta'}(k))$. Then, there exists an $S$-solution for $\Pi_{\Delta}(k)$. 
\end{lemma}
\begin{proof}
    For each edge $e$ incident to a node $v$ in $S$, let $L_e(v)$ be its label-set. 
    We start by transforming $L_e(v)$ into a set of colors that satisfies some desirable properties. Let $C_e(v) = \cup_{\ell(C) \in L_e(v)} C$. Observe that, for two nodes $v$ and $u$ that are incident to the same edge $e$, $C_e(v) \cap C_e(u) = \emptyset$. In fact, by the definition of the edge constraint of $\Pi'$ (see \Cref{def:lift}) it holds that if $\ell(C_1) \in L_e(v)$ and $\ell(C_2) \in L_e(u)$, then $\{\ell(C_1),\ell(C_2)\}$ must be contained in the edge constraint of $\Pi_{\Delta'}(k)$, which in turn requires $C_1$ and $C_2$ to be disjoint. 

    For each node $v$ in $S$, we define a bipartite graph $H = (U_H \cup V_H, E_H)$ as follows. 
    Let $e_1,\ldots, e_{\Delta}$ be the edges incident to $v$, taken in an arbitrary order.
    The set $U_H$ is defined as $U_H := \{u_1,\ldots,u_k\}$. The set $V_H$ is defined as $V_H := \{v_1,\ldots,v_{\Delta}\}$. The set of edges $E_H$ is defined as follows. There is an edge between $u_i \in U_H$ and $v_j \in V_H$ if and only if $i \notin C_{e_j}(v)$. For a subset of colors $C = \{c_{i_1},\ldots,c_{i_{|C|}}\}$, let $N(C)$ be the union of the neighbors of nodes $v_{i_1},\ldots,v_{i_{|C|}}$.
    Assume, for a contradiction, that the following holds.
    \[
        \forall C \subseteq \mathcal{C}, \; |C| \le |N(C)|
    \]
    By Hall's marriage theorem \cite{hall}, this condition implies that there exists a matching in $H$, where all nodes in $U_H$ are matched. Consider a subset $E'$ of $\Delta'$ edges incident to $v$ that include the edges corresponding to matched nodes in $V_H$ (note that this subset exists since $k \le \Delta'$). We claim that, when considering the chosen edges $E'$, the constraint of $\Pi'$ is not satisfied on $v$. Suppose, for a contradiction, that there exists a configuration $\ell(C)^{\Delta'-x} \s \X^x$, where $x = |C| - 1$, valid for these edges. This implies that there exist $\Delta'-x$ edges $e \in E'$ incident to $v$ satisfying $\ell(C) \in L_e(v)$. By the definition of $C_e(v)$, this implies that, in $E'$, there are at least $\Delta' - x = \Delta' - |C| + 1$ edges $e$ satisfying that $C \subseteq C_e(v)$.
    However, by the construction of $E'$, there are at least $|C|$ edges $e$ in $E'$ satisfying that $C_e(v)$ misses at least one color of $C$. Thus, at most $\Delta' - |C|$ edges $e \in E'$ satisfy that $C \subseteq C_e(v)$, reaching a contradiction.

    Hence, we get that, for all nodes $v$, there exists a set $C \subseteq \mathcal{C}$ satisfying that $|C| \ge |N(C)| + 1$.
    This implies that there exists a set $C \subseteq \mathcal{C}$ satisfying that at most $|C|-1$ edges $e$ of $v$ satisfy $C \not\subseteq C_e(v)$, and hence that we can assign the configuration $\ell(C)^{\Delta-x} \s X^x$, where $x = |C| - 1$, to~$v$.
\end{proof}

\begin{lemma}\label{lem:col-supergraph}
Let $S$ be a subset of nodes of $G$, and let $k \le \Delta'$ be an integer. Assume there exists an $S$-solution for $\Pi_{\Delta}(k)$. Then, the subgraph induced by the nodes in $S$ can be colored with $2k$ colors.
\end{lemma}
\begin{proof}
For each $v \in S$, let $\ell(C_v)^{\Delta'-x_v} \s \X^{x_v}$, where $x_v = |C_v| - 1$, be the configuration of $v$ in the $S$-solution. Observe that, if each node $v$ picks an arbitrary color from $C_v$, then we obtain a coloring of the nodes of $S$ satisfying that the only monochromatic edges in the subgraph induced by nodes in $S$ are edges that are labeled $\X$ on at least one side. We prove that, at the cost of doubling the amount of colors and using a more careful assignment, then we can properly color the subgraph induced by nodes in $S$.
More in detail, we provide a function that maps each node $v$ into a color from $C'_v := \{c_{i,1},c_{i,2} \mid c_i \in C_v\}$, such that we obtain a proper coloring of the nodes in $S$. 

Let $G_X$ be the graph obtained as follows. Start from the subgraph induced by the nodes in $S$ and throw away all edges $e$ that are labeled with a configuration $\{\ell_1,\ell_2\}$ satisfying that both $\ell_1$ and $\ell_2$ are not $\X$.
We prove that we can provide an ordering $\mathcal{O} = (v_{1},v_{2},\ldots)$ of the nodes that satisfies the following property. For each node $v_{i}$, let $V_{\ge i} = \{ v_{j} \mid j \ge i\}$, and let $G_{\ge i}$ be the subgraph of $G_X$ induced by nodes in $V_{\ge i}$. The degree of $v_i$ in $G_{\ge i}$ is bounded by $2|C_{v_i}| - 1$.

Suppose we already constructed the prefix of the ordering $(v_1,\ldots,v_{i-1})$. We show that a node $v_i$ that satisfies the above property exists. For each node $v \in V_{\ge i}$, let $d_i(v)$ be the degree of node $v$ in $G_{\ge i}$.
Let $E(G_{\ge i})$ be the edges of $G_{\ge i}$. We start by proving that $|E(G_{\ge i})| \le \sum_{v \in G_{\ge i}} (|C_{v}| - 1)$. For each edge to be in $E(G_{\ge i})$, it needs to be labeled $\X$ on at least one side. Hence, $|E(G_{\ge i})|$ is upper bounded by the amount of $\X$ in $G_{\ge i}$, which in turn is bounded by $\sum_{v \in G_{\ge i}} (|C_{v}| - 1)$. We will use this property in the following.

Suppose, for a contradiction, that all nodes $v \in V_{\ge i}$ satisfy $d_i(v) \ge 2|C_{v}|$.
Then, the following holds.
\begin{align*}
    \forall v \in G_{\ge i} ~ 2|C_{v}| \le d_i(v) &\Longrightarrow \sum_{v \in G_{\ge i}} 2|C_{v}| \le \sum_{v \in G_{\ge i}} d_i(v) \\
    &\Longrightarrow \sum_{v \in G_{\ge i}} 2|C_{v}| \le 2 |E(G_{\ge i})| 
    &\Longrightarrow \sum_{v \in G_{\ge i}} |C_{v}| \le \sum_{v \in G_{\ge i}} (|C_{v}| - 1),
\end{align*}
which is a contradiction if $G_{\ge i}$ contains at least one node. Hence, there exists a node $v_i$ that satisfies  $d_i(v) \le 2|C_{v_i}| - 1$, as required. Hence, we can construct the ordering $\mathcal{O}$, as desired.

We now prove that we can process the nodes in the ordering that is the reverse of $\mathcal{O}$, and color each node $v$ with a color from the set $C'_v$ satisfying that the graph induced by nodes in $S$ is properly colored.
Assume that the nodes $v_{i+1},v_{i+2},\ldots$ are already colored. We show that we find a color for $v_i$ such that $G_{\ge i}$ is properly colored. Node $v_i$ has at most $2|C_{v_i}| - 1$ neighbors that are already colored, and it has $2|C_{v_i}|$ colors available. Hence, $v_i$ can pick a unused color from $C'_{v_i}$.
This implies that, after processing all nodes, we obtain a proper coloring of $G_X$ where each node $v$ gets a color from $C'_v$, which implies a proper coloring for the subgraph induced by nodes in $S$.
\end{proof}

\section{Arbdefective Colored Ruling Sets}\label{sec:acbrs}
In this section, we prove the following theorem.
\begin{restatable}{theorem}{acbrs}\label{thm:acbrs}
Let $\mathcal{G}$ be the family of $\Delta$-regular graphs of girth $\Omega(\log_\Delta n)$.
     Let $\Delta$ be the degree of the support graph $G$, and let $\Delta'$ be the degree of the input graph $G'$.
     Assume that the support graph is from $\mathcal{G}$, and that the input graph has degree $\Delta'$ satisfying $\Delta \ge 3 \Delta'$.
     Let $\Bar{\Delta} = \min\{\Delta', \eps \Delta / \log \Delta \} / 2^{c \cdot \beta}$, for some small-enough constant $\eps$, and some large-enough constant $c$.
     Let $\alpha\ge 0$, $c \ge 1$, $\beta \ge 1$ be integers satisfying $(\alpha+1)c \le \Bar{\Delta}$ and $\beta < \Delta'$.
     Then, the $\alpha$-arbdefective $c$-colored $\beta$-ruling set problem, in the Supported LOCAL model, requires $T_{\mathrm{det}} =\Omega(\min\{ \beta(\frac{\Bar{\Delta}}{(\alpha + 1)c})^{1 / \beta}), \log_\Delta n \})$ deterministic rounds and $T_{\mathrm{rand}} = \Omega(\min\{ \beta(\frac{\Bar{\Delta}}{(\alpha + 1)c})^{1 / \beta}), \log_\Delta \log n \})$ randomized rounds, for all $\beta$ satisfying $T_{\mathrm{det}} - \beta = \Omega(T_{\mathrm{det}})$ in the deterministic case, and $T_{\mathrm{rand}} - \beta = \Omega(T_{\mathrm{rand}})$ in the randomized case.
\end{restatable}

\subsection{Problem Definition}
In order to prove \Cref{thm:acbrs}, we consider a family of problems that we will later show to be strongly related with $\alpha$-arbdefective $c$-colored $\beta$-ruling sets.

\paragraph{A family of problems in the black-white formalism.}
We define a problem $\Pi_\Delta(c,\beta)$ in the black-white formalism as follows.
\begin{definition}[The problem $\Pi_\Delta(c,\beta)$]
Let $\mathcal{C} := \{1,\ldots,c\}$, let $\beta \ge 0$ be an integer, and let $\Sigma_{\Pi_\Delta(c,\beta)} := \{ \X \} \cup \{ \ell(C) \mid C \subseteq \mathcal{C} \text{ and } C \neq \emptyset \} \cup \{\P_i,\U_i \mid 1 \le i \le \beta\}$.
If $\beta = 0$, the problem $\Pi_\Delta(c,\beta)$ is defined as the problem $\Pi_\Delta(c)$ of \Cref{def:arbdef}. For $\beta \ge 1$, the problem $\Pi_\Delta(c,\beta)$ is defined via the following white and black constraints.
\begin{equation*}
	\begin{aligned}
		\begin{aligned}
			\Pi^W_\Delta(c,\beta)\text{:}\\
	&\ell(C)^{\Delta-x} \s \X^x,&& \text{where} \; x=|C|-1, \\ &&&\text{for all }  C \subseteq \mathcal{C}, C \neq \emptyset\\
        & \P_i \s \U_i^{\Delta - 1}, && \text{for all $1 \le i \le \beta$}
 		\end{aligned}
   \qquad
		\begin{aligned}
  			\Pi^B_\Delta(c,\beta)\text{:}\\
	&\ell(C_1) \s \ell(C_2),&& \text{ for all $C_1, C_2$ s.t. }\\&&&\text{ } C_1 \cap C_2 = \emptyset, \\ &&& \text{ where } C_1, C_2 \neq \emptyset\\
    &\X \s \L,&& \text{ for all $\L \in \Sigma_{\Pi_\Delta(c,\beta)} $}\\
        &\P_i \s \U_j,&& \text{ for all $1 \le j < i \le \beta$} \\
        &\P_i \s \ell(C),&& \text{ for all $1 \le i \le \beta$},\\&&&\text{ for all $C \subseteq \mathcal{C}, C \neq \emptyset$} \\
        &\U_i \s \U_j,&& \text{ for all $1 \le i, j \le \beta$}\\
		\end{aligned}
	\end{aligned}
\end{equation*}
\end{definition}
An example of black diagram of the probles in the family is shown in \Cref{fig:diagram-acrs}.
  \begin{figure}[b]
        \centering
        \includegraphics[width = 0.27\textwidth]{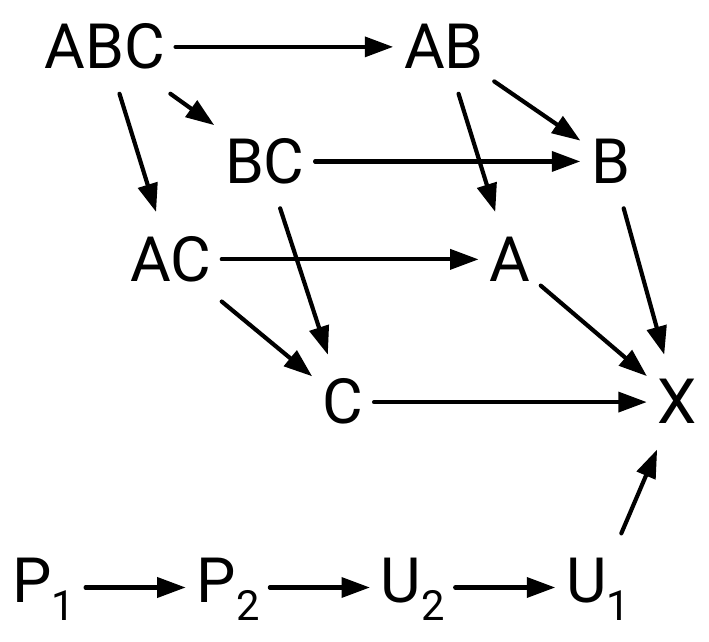}
        \caption{Black diagram of $\Pi$, in the case where $\mathcal{C}$ contains $3$ colors denoted with $\A$, $\B$, and $\C$, and $\beta = 2$.}
        \label{fig:diagram-acrs}
    \end{figure}

Observe that the problem $\Pi_\Delta(c,\beta)$ is defined very similarly as the problem $\Pi_\Delta(c)$ of \Cref{def:arbdef}. In particular, $\Pi_\Delta(c,\beta)$ can be obtained from $\Pi_\Delta(c)$ by performing the following operations.
\begin{itemize}
    \item On the node constraint, we add the configuration $\P_i \s \U_i^{\Delta-1}$ for each $1 \le i \le \beta$.
    \item On the edge constraint, we make $\P_i$ and $\U_i$ compatible with all the labels of $\Pi_\Delta(c)$.
    \item Additionally, on the edge constraint, we make $\U_i$ compatible with $\U_j$ for all pairs $(i,j)$, and we make $\P_i$ compatible with $\U_j$ only if $i > j$.
\end{itemize}
Intuitively, a valid solution of $\Pi_\Delta(c,\beta)$ can be obtained as follows.
\begin{itemize}
    \item Select a subset of nodes $S$ that satisfies that all nodes in $V \setminus S$ have a node in $S$ at distance at most $\beta$, and solve $\Pi_\Delta(c)$ on them.
    \item All nodes in $V \setminus S$ can use labels $\P_i$ and $\U_i$ to point to a node in $S$ at distance at most $\beta$.
\end{itemize}

\paragraph{What is known about these problems.}
In \cite{Balliu0KO22}, it has been shown that $\Pi_\Delta(c,\beta)$ is strictly related to $\alpha$-arbdefective $c$-colored $\beta$-ruling sets, and in particular that a solution for $\alpha$-arbdefective $c$-colored $\beta$-ruling sets can be converted in $\beta$ rounds into a solution for $\Pi_\Delta((\alpha+1)c,\beta)$, implying that $\Pi_\Delta((\alpha+1)c,\beta)$ is at most $\beta$ rounds harder than $\alpha$-arbdefective $c$-colored $\beta$-ruling sets. 
\begin{lemma}[Theorem 1.5 (in the ArXiv version) of \cite{Balliu0KO22}]\label{lem:arbdef-to-pi}
In the LOCAL model, given a solution for $\alpha$-arbdefective $c$-colored $\beta$-ruling sets, it is possible to solve $\Pi_\Delta((\alpha+1)c,\beta)$ in $\beta$ rounds. 
\end{lemma}
Moreover, in \cite{Balliu0KO22} it is also shown how, in the round elimination framework, problems $\Pi_\Delta(k,\beta)$ with different parameters are related.
\begin{lemma}[Lemma 6.1, Lemma 6.13, Lemma 8.6, and Corollary 8.8 (in the ArXiv version) of \cite{Balliu0KO22}]\label{lem:acb-sequence}
Let $t:= \lfloor \eps \beta(\frac{k}{(\alpha+1)c})^{1/\beta} \rfloor$, for some small-enough constant $\eps$, and any integer $1 \le k < \Delta$. 
Then, there exists a lower bound sequence $\Pi_1,\ldots,\Pi_t$, where $\Pi_1 = \Pi_\Delta((\alpha+1)c,\beta)$ and $\Pi_t$ can be related to $\Pi_\Delta(k,\beta)$.
\end{lemma}

\subsection{A Lower Bound for the Supported LOCAL Model}
In order to prove \Cref{thm:acbrs}, we follow the same strategy as in the case of $x$-maximal $y$-matchings and arbdefective colorings, that is, we prove the following lemma, that combined with \Cref{thm:approach-hypergraphs}, \Cref{lem:acb-sequence}, and \Cref{lem:arbdef-to-pi}, gives \Cref{thm:acbrs}.
\begin{lemma}\label{lem:acb-lift-no-sol}
Let $k := \lfloor \min\{\Delta', \eps \Delta / \log \Delta \}/ 2^{c \cdot \beta} \rfloor$, for a small-enough constant $\eps$, and a large-enough constant $c$.
For any $n$ and $\Delta$ such that $n \ge \Delta^2$, there exists a $\Delta$-regular graph $G$ with at least $n / \Delta^2$ nodes and at most $n$ nodes, of girth at least $\epsilon' \log_\Delta n$ (for some absolute constant $\varepsilon'$), and where,  assuming $\Delta \ge 3 \Delta'$,  $\Pi' := \mathrm{lift}_{\Delta,2}(\Pi_{\Delta'}(k,\beta))$ has no non-bipartite solution.
\end{lemma}
In order to prove this lemma, we use exactly the graph family $\mathcal{G}$ given by \Cref{lem:graph-family}. Hence, in the following, let $G \in \mathcal{G}$. We need to prove that, on $G$, $\Pi'$ is not non-bipartitely solvable.

For this purpose, consider the following family of problems. The problem $\bar{\Pi}_{\Delta',x}(k,\beta)$ has the same edge constraint as $\mathrm{lift}_{\Delta,2}(\Pi_{\Delta'}(k,\beta))$, and the node constraint requires that each node satisfies the node constraint of $\mathrm{lift}_{\Delta,2}(\Pi_{\Delta' - y}(k,\beta))$, for some $y \in \{1,\ldots,x\}$, where different nodes may use different values of $y$.
We prove the following statement (recall the notion of $S$-solution defined in \Cref{def:s-solution}). 
\begin{lemma}\label{lem:acb-recursive}
    Let $S$ be a subset of nodes of $G$. Assume there exists an $S$-solution for $\bar{\Pi}_{\Delta',x}(k,\beta)$ satisfying that, for all edges $e = \{u,v\}$ such that $u \in S$ and $v \notin S$, the label assigned by $u$ on $e$ does not contain any label $\P_i$, for $1 \le i \le \beta$.
    Then, there exists a subset $S'$ of nodes of $G$ satisfying $|S'| \ge |S|/4$ such that there exists an $S'$-solution for $\bar{\Pi}_{\Delta',x+1}(2k,\beta-1)$ on $G$ satisfying that, for all edges $e = \{u,v\}$ such that $u \in S'$ and $v \notin S'$, the label-set assigned by $u$ on $e$ does not contain any label $\P_i$, for $1 \le i \le \beta -1$.
\end{lemma}
We now show that \Cref{lem:acb-recursive} implies \Cref{lem:acb-lift-no-sol}. Later, we will prove \Cref{lem:acb-recursive}.
Assume, for a contradiction, that $\mathrm{lift}_{\Delta,2}(\Pi_{\Delta'}(k,\beta))$ is solvable. Note that $\mathrm{lift}_{\Delta,2}(\Pi_{\Delta'}(k,\beta)) = \bar{\Pi}_{\Delta',0}(k,\beta)$.
By applying \Cref{lem:acb-recursive} recursively for $\beta$ times, we obtain that if there exists a $V$-solution on $G = (V,E)$ for $\mathrm{lift}_{\Delta,2}(\Pi_{\Delta'}(k,\beta))$, then there exists a subset $S$ of nodes of size $n / 4^\beta$ satisfying that an $S$-solution for $\bar{\Pi}_{\Delta',\beta}(2^\beta k,0)$ exists on $G$. 
Note that, if a node has a labeling that is valid for $\mathrm{lift}_{\Delta,2}(\Pi_{\Delta'-x}(2^\beta k,0))$, then it has a valid labeling also for $\mathrm{lift}_{\Delta,2}(\Pi_{\Delta' - x'}(2^\beta,0))$, for any $x' > x$. Hence $\bar{\Pi}_{\Delta',\beta}(2^\beta k,0)$ is equivalent to $\mathrm{lift}_{\Delta,2}(\Pi_{\Delta'-\beta}(2^\beta k,0))$, which in turn is equivalent to $\mathrm{lift}_{\Delta,2}(\Pi_{\Delta'-\beta}(2^\beta k))$, that is, the lift of the problem defined in \Cref{def:arbdef}.  Recall that $k$ is defined as $k := \lfloor \min\{\Delta', \eps \Delta / \log \Delta \}/ 2^{c \cdot \beta} \rfloor$.
Observe that
\[
2^\beta k = 2^\beta \cdot \lfloor \min\{\Delta', \eps \Delta / \log \Delta \}/ 2^{c \cdot \beta}\rfloor \le 2^\beta \cdot \Delta' / 2^{c \cdot \beta} = \frac{\Delta'}{2^{c\beta - \beta}},
\]
that, for large-enough $c$, by the assumption that $\beta < \Delta'$, is at most $\Delta' - \beta$.
Hence, by applying \Cref{lem:subgraph-coloring}, we obtain that it is possible to color a fraction $1 / 4^\beta$ of the the nodes of $G$ with $2 \cdot 2^\beta \cdot k= 2^{\beta+1} \cdot k$ colors.  Thus, by picking $c$ large enough, if there exists a solution for $\mathrm{lift}_{\Delta,2}(\Pi_{\Delta'}(k,\beta))$ on $G$, then it is possible to color a fraction $1 / 4^\beta$ of the the nodes of $G$ with the following amount of colors:
\[
\min\{\Delta', \eps \Delta / \log \Delta \} / 2^{c' \beta} \le \frac{\eps \Delta}{ 2^{c' \beta} \log \Delta},
\]
for any chosen constant $c'$. By assumption, the graph $G$ satisfies that the largest independent set in $G$ has size $\gamma n \frac{\log \Delta}{\Delta}$ for some constant $\gamma > 0$.
Observe that this property implies that, if we consider the subgraph $G_S$ induced by an arbitrary fraction $1 / 4^\beta$ of the nodes, the chromatic number of $G_S$ is lower bounded by:
\[
\frac{n}{4^\beta} \cdot \frac{\Delta}{\gamma \cdot n \cdot \log \Delta} = \frac{\Delta}{\gamma \cdot 2^{2\beta} \cdot \log \Delta}.
\]
Observe that, by picking $\eps$ small enough, and $c'$ large enough, we obtain a contradiction. Hence, there is no solution for $\mathrm{lift}_{\Delta,2}(\Pi_{\Delta'}(k,\beta))$ on $G$, implying \Cref{lem:acb-lift-no-sol}. In the rest of the section, we prove \Cref{lem:acb-recursive}.

\begin{proof}
    Assume there is an $S$-solution to $\bar{\Pi}_{\Delta',x}(k, \beta)$ satisfying the requirements of \Cref{lem:acb-recursive}. The goal is to modify the label-sets of the $S$-solution in order to get rid of the labels $\P_\beta$ and $\U_\beta$.
    We consider all nodes $u$ that have at least one incident edge with a label-set that contains $\P_\beta$ or $\U_\beta$. 
    Such nodes could be of three possible types.
    \begin{itemize}
        \item \textbf{Type 1}: all edges incident to $u$ have $\U_\beta$ in their label-set, and the number of edges incident to $u$ that have $\P_\beta$ in their label-set is greater than $\Delta-\Delta+1$. We will prove that there are at most $3|S|/4$ of these nodes, and we will define $S'$ as the nodes of $S$ without these ones.
        \item \textbf{Type 2}: all edges incident to $u$ have $\U_\beta$ in their label-set, and the number of incident edges having $\P_\beta$ in their label-sets is at most $\Delta-\Delta'$. We will prove that, at the cost of increasing the number of colors by $k$, it is possible to assign, to the edges incident to these nodes, label-sets not containing $\U_i$ nor $\P_i$ for any $i$, such that the constraints of the problem are satisfied. In other words, we can assign label-sets containing only sets of colors and $\X$, in such a way that the constraints are satisfied on the subset of nodes $S'$ that remain, that is, on all the nodes that are not of type 1.
        \item \textbf{Type 3}: there is an edge $e$ incident to $u$ whose label-set does not contain $\U_\beta$. We will prove that, 
        if $u$ satisfies the node constraint of $\mathrm{lift}(\Pi_{\Delta' - y}(k,\beta))$ for some $y \in \{1,\ldots,x\}$, then it also satisfies the node constraint of $\mathrm{lift}(\Pi_{\Delta' - y - 1}(k,\beta - 1))$.
        In other words, at the cost of decreasing by $1$ the considered degree, we can always pick a labeling for $u$ that does not use $\P_\beta$ or $\U_\beta$. Hence, we can solve $\bar{\Pi}_{\Delta',x+1}(k, \beta - 1)$ on $u$.
    \end{itemize}
    We do not modify the solution of the other nodes. In fact, the other nodes already satisfy the node constraint of $\bar{\Pi}_{\Delta',x}(k, \beta-1)$, and all the node configurations allowed by $\bar{\Pi}_{\Delta',x}(k, \beta-1)$ are also allowed by $\bar{\Pi}_{\Delta',x+1}(k, \beta-1)$.

    \paragraph{Type 1 nodes.} Let $u$ be a node of type 1. 
    By assumption, for all edges $e=\{u,v\}$, if $u \in S$ and $v\notin S$, the label-set assigned by $u$ to $e$ does not contain the label $\P_i$ for any $i$. In particular, this holds for $\P_\beta$. Hence, all the edges incident to $u$ whose label-sets contain $\P_\beta$ have the other endpoint in $S$. 
    Recall that the edge constraint of $\bar{\Pi}_{\Delta',x}(k, \beta)$ requires that each pair of label-sets assigned to an edge satisfies that, for every possible choice of labels over the pairs, the obtained pair is in the edge constraint of $\Pi_{\Delta'(k,\beta)}$. Thus, for each edge of $S$ it holds that, if the label-set assigned to one half-edge contains $\P_\beta$, then the label-set assigned to the other half-edge does not.
    Hence, the number of half-edges with the label $\P_\beta$ in their label-set in the subgraph of $G$ induced by $S$ is at most $|S|\Delta/2$, as there are at most $|S|\Delta$ half-edges between nodes in $S$, and the two half-edges of the same edge cannot both have $\P_\beta$ in their label-set. Since $u$ is a type 1 node, by assumption it has at least $\Delta-\Delta'$ edges with label-sets containing $\P_\beta$, and these edges must be part of the subgraph induced by $S$.  It follows that type 1 nodes are at most $\frac{|S|\Delta}{2}\cdot \frac{1}{\Delta-\Delta'}$. Recall that, by assumption, $\Delta \ge 3 \Delta'$. Thus, type 1 nodes are at most $3|S|/4$. The set $S'$ is defined as the set of nodes that remain in $S$ after removing type 1 nodes. Observe that, since all edges incident to $u$ contain $\U_\beta$, and since $\U_\beta$ is not compatible with $\P_i$ for any $i$, we get that there is no edge $e = \{u,v\}$ such that $u \in S'$ and $v \notin S'$, satisfying that the label-set assigned by $u$ on $e$ contains $\P_i$ for some $i$, as required.

    \paragraph{Type 2 nodes.} Let $u$ be a node of type 2. We provide an entirely new assignment of label-sets for the half-edges incident to $u$, that uses label-sets containing only subsets of $\{k+1,\ldots,2k\}$, and $\X$. We split the half-edges of $u$ into two groups: a half-edge is a type-$U$ edge if its label-set does not contain $\P_\beta$, while it is a type-$P$ edge if its label-set contains $\P_\beta$. An edge is a $U$-$U$ edge if it is composed of two type-$U$ half-edges of nodes of type 2. 
    Observe that nodes of type 2 can only be neighbors via $U$-$U$ edges, since all half-edges of type 2 nodes contain $\U_\beta$, which is not compatible with $\P_\beta$. We thus get that, if we assign to all half-edges of $u$ label-sets containing only subsets of $\{k+1,\ldots,2k\}$, and $\X$, the edge constraints are trivially satisfied on $P$-edges, and on $U$-edges that are not part of $U$-$U$ edges, since these labels are compatible with all the labels that are not subsets of $\{k+1,\ldots,2k\}$, that is, all labels that could possibly be present on the other side of such edges.

    By the definition of type 2 nodes, there are at least $\Delta'$ $U$-edges incident to $u$. Note that, by right-closedness, if a label-set does not contain $\P_\beta$, then it cannot contain any $\P_i$ for any $i < \beta$ as well, due to the fact that if some configuration $\{\ell,\P_i\}$ is allowed on the edges, then $\{\ell,\P_\beta\}$ is also allowed. We thus get that $U$-edges do not contain any label $\P_i$ for any $i$.
    Recall that $u$ satisfies the node constraint of $\mathrm{lift}(\Pi_{\Delta' - y}(k,\beta))$, for some $y \in \{1,\ldots,x\}$. Combining this with the fact that $U$-edges do not contain any $\P_i$, we get that, for any choice of $\Delta'-y$ edges over the $\Delta'$ $U$-edges (which is at least one choice), there exists a choice of labels from the label-sets of the form $\ell(C)^{\Delta'-y - (|C|-1)} \s \X^{|C|-1}$, where $C$ is a subset of $\{1,\ldots,k\}$. For each $U$-edge of $u$ with assigned label-set $L$, we assign the new label-set $\{\ell(\{c + k \mid c \in C\}) \mid \ell(C) \in L\} \cup \{\X\}$, or in other words, we discard $\U_i$ labels, and we shift each color by $k$. To all other edges (the $P$-edges) we assign the same label-set, which is the union of all the label-sets assigned to the $U$-edges. We obtain that the constraint of $\mathrm{lift}(\Pi_{\Delta' - y}(2k,\beta))$ is satisfied on $u$. Moreover, since before shifting the colors by $k$ the edge constraint was satisfied, $U$-$U$ edges still satisfy the edge constraint (because each color is shifted by the same amount). As already discussed, the edge constraint are still satisfied on all the other edges.

    \paragraph{Type 3 nodes.} Let $u$ be a node of type 3, and let $e$ be a half-edge incident to $u$ whose label-set does not contain $\U_\beta$.
    By assumption, $u$ satisfies the node constraint of $\mathrm{lift}(\Pi_{\Delta' - y}(k,\beta))$, for some $y \in \{1,\ldots,x\}$.
    We show that, either we can just discard all $\P_\beta$ and $\U_\beta$ to satisfy $\mathrm{lift}(\Pi_{\Delta' - y}(k,\beta-1))$, or, for any choice of $\Delta' - y - 1$ half-edges incident to $u$, there exists a choice over the label-sets assigned to these half-edges that is in the node constraint of $\Pi_{\Delta' - y - 1}(k,\beta - 1)$, and hence we can discard all $\P_\beta$ and $\U_\beta$ to satisfy $\mathrm{lift}(\Pi_{\Delta' - y - 1}(k,\beta-1))$. Let $M$ be an arbitrary set of $\Delta' - y - 1$ half-edges.
    
    Suppose $e \notin M$. Let $M' = M \cup \{e\}$. We start by showing that, over the label-sets assigned to the half-edges $M'$ there exists a configuration that we can pick that does not use $\P_\beta$ and $\U_\beta$, and that if it uses $\P_i$ for some $i < \beta$, then $\P_i$ is not on $e$. Observe that, if a configuration $\{\P_i,\ell\}$ is allowed by the edge constraint of $\Pi_{\Delta'(k,\beta)}$, then the configuration $\{\U_\beta,\ell\}$ is also allowed. Thus, by right-closedness, since the label-set of $e$ does not contain $\U_\beta$, the label-set of $e$ does also not contain any label $\P_i$ for any $i$. Hence, over the half-edges $M'$, we can pick a configuration that is valid for $\Pi_{\Delta' - y}(k,\beta)$ that does not use $\P_i$ on $e$ (for any $i$), and that does not use $\U_\beta$ at all (since it is not present on $e$). We consider two cases separately.
    \begin{itemize}
        \item The configuration for $M'$ is of the form $\P_i \s \U_i^{\Delta' - y - 1}$, where $i < \beta$, and where $\P_i$ is not on $e$. Then, on $M$, we can pick the configuration $\P_i \s \U_i^{\Delta' - y - 2}$.
        \item The configuration for $M'$ is of the form $\ell(C)^{\Delta'-y - (|C|-1)} \s \X^{|C|-1}$. Since, by right-closedness, all label-sets contain $\X$, we get that on $M$ we can pick the configuration $\ell(C)^{\Delta'-y - (|C|-1) - 1} \s \X^{|C|-1}$.
    \end{itemize}
    Suppose now that $e \in M$. We consider two cases separately.
    \begin{itemize}
        \item There exists an edge $e'\notin M$ such that, on $M' = M \cup \{e'\}$, we can pick a configuration of the form $\ell(C)^{\Delta'-y - (|C|-1)} \s \X^{|C|-1}$, or of the form $\P_i \s \U_i^{\Delta' - y - 1}$ (which must satisfy $i < \beta$, since $e$ does not contain $\U_\beta$ nor $\P_\beta$) in which $\P_i$ is not picked from $e'$. Similarly as before, we obtain that we can pick a configuration for $M$.
        \item For all edges $e'\notin M$, on $M' = M \cup \{e'\}$, the only valid configurations that can be picked is $\P_i \s \U_i^{\Delta' - y - 1}$ for some $i$ (which must satisfy $i < \beta$, since $e$ does not contain $\U_\beta$), where $\P_i$ is picked from the label-set of $e'$. We get that all edges that are not in $M$ (which are $\Delta - (\Delta' - y - 1)$) contain at least one $\P_i$, for $i < \beta$, which, by right-closedness, it implies that these edges contain $\U_j$ for all $j$. Consider an arbitrary choice of $\Delta'-y$ edges of $u$: we get that we can pick $\P_i \s \U_i^{\Delta' - y - 1}$, for some $i < \beta$. Hence, if we discard all $\P_\beta$ and $\U_\beta$, node $u$ satisfies the node constraint of $\Pi_{\Delta' - y}(k,\beta-1)$.
    \end{itemize}

\end{proof}

\section{Conclusions and Open Questions}
In this work, we have shown that essentially all lower bounds for the LOCAL model proved via round elimination hold in the Supported LOCAL model as well. However, there are few exceptions, that we leave as open questions.

    Our lower bounds for arbdefective colored ruling sets are only tight for constant values of $\beta$, and we leave as an open question to determine whether this can be improved.

    In \cite{mm-hypergraphs}, interesting lower bounds for the LOCAL model have been shown for problems on hypergraphs. These problems have not been tackled in our work, and we leave as an open question to determine their complexity in the Supported LOCAL model.

\urlstyle{same}
\bibliographystyle{alpha}
\bibliography{biblio}

\appendix

\clearpage

\section{Example of Problems in the Black-White Formalism}\label{sec:example-black-white}
The maximal matching problem on bipartite $2$-colored graphs can be encoded in the black-white formalism by the following constraints (see \Cref{fig:mm-black-white} for an example of assignment of labels).

\begin{equation*}
	\begin{aligned}
		\begin{aligned}
			\C_W\text{:}\\
	&\M \s \O^{\Delta-1}\\
        & \P^{\Delta}
 		\end{aligned}
   \qquad
		\begin{aligned}
			\C_B\text{:}\\
	&\M \s [\O\P]^{\Delta-1}\\
        & \O^{\Delta}
 		\end{aligned}
	\end{aligned}
\end{equation*}

The label $\M$ on an edge $e$ indicates that $e$ is in the matching, while the label $\O$ on $e$ indicates that $e$ is \emph{not} in the matching ($\M$ stands for ``matched'' while $\O$ stands for ``other''). First of all, notice that the black and white configurations are such that, for each (black or white) node $v$, $\M$ is outputted on at most one incident edge of $v$. This guarantees that there is never more than one edge incident to a node that is in the matching. We now go through the white and black configurations separately, and show that the requirements of the maximal matching problem are satisfied on both white and black nodes.

A white node $v$ is matched if it outputs the configuration $\M \s \O^{\Delta-1}$, which indicates that exactly one incident edge to $v$ is in the matching while the others are not. Then, a white node $v$ is unmatched if it outputs the configuration $\P^{\Delta}$. By the definition of the maximal matching problem, we want to satisfy that, if a node is not matched, then all its neighbors are already matched with someone else (otherwise the matching would not be maximal). Let $\{u,v\}$ be any edge labeled $\P$ where $u$ is a black node and $v$ is a white node. In order to satisfy maximality for the white nodes, the black node $u$ must be matched. In fact, the only black configuration that contains the label $\P$ is the one that outputs $\M$ on exactly one incident edge, and thus maximality is satisfied on the white nodes. 

A black node $u$ is matched if it outputs a configuration in $\M \s [\O\P]^{\Delta-1}$, guaranteeing that in this case exactly one incident edge of $u$ is in the matching. Moreover, this configuration says that the edges incident to $u$ that are not in the matching can be labeled either $\P$ or $\O$. This indicates that node $u$ can accept pointer-labels $\P$ from white nodes, but it can also be neighbor to white nodes that are matched with some other node different from $u$, and hence $u$ accepts label $\U$ as well. 
Then, a black node $u$ that is not matched outputs the configuration $\O^{\Delta}$, and this satisfies the maximality constraint on black nodes: indeed, all edges $\{u,v\}$ having label $\O$ indicate that the white node $v$ is a matched node, since the only possible white configuration containing $\O$ is $\M \s \O^{\Delta-1}$.

Therefore, the provided description in the black-white formalism satisfies, at all nodes, the packing and covering constraints of the maximal matching problem.

 \begin{figure}
        \centering
        \includegraphics[width = 0.5\textwidth]{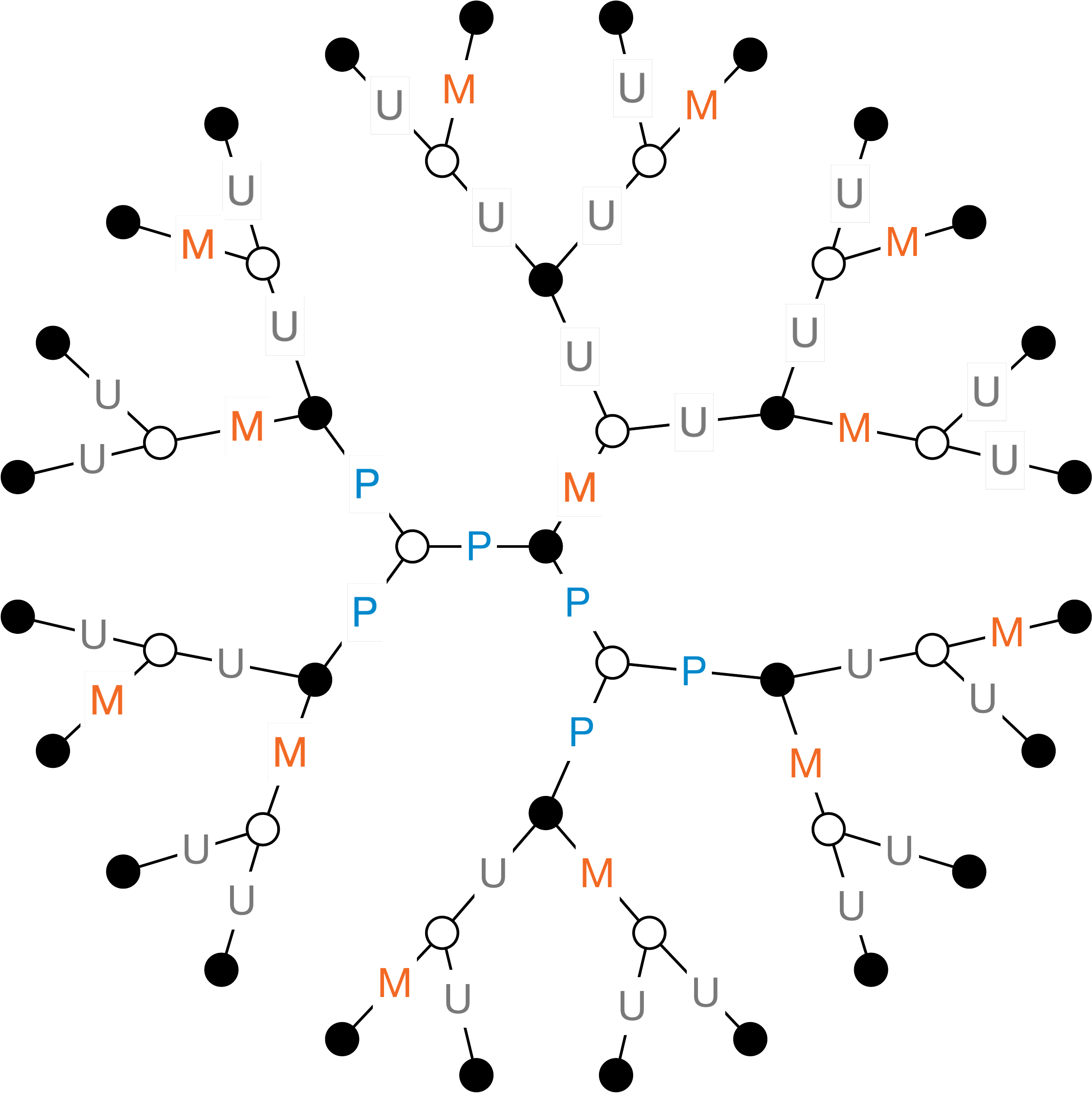}
        \caption{An example of a solution to the maximal matching problem in the black-white formalism.}
        \label{fig:mm-black-white}
    \end{figure}

The black diagram of the problem contains only the directed edge $(\P,\O)$. The reason why this edge is in the diagram is that, for any black configuration containing $\P$, we can replace an arbitrary amount of $\P$ with $\O$ and still obtain a configuration in the black constraint. We get that $\O$ is at least as strong as $\P$. Observe that no other pair of labels satisfies a similar property.

\section{Round Elimination in the Supported LOCAL Model}\label{sec:re}
In this section, we show that the round elimination technique works in the Supported LOCAL model as well. We start by defining two functions, $\re$ and $\rere$, that take as input a problem $\Pi$ in the black-white formalism, and output a problem $\Pi'$ in the black-white formalism. Then, the function $\mathrm{RE(\Pi)}$ used in the rest of the paper is defined as $\mathrm{RE}(\Pi) := \rere(\re(\Pi))$. Consider a problem $\Pi=(\Sigma, \C_W, \C_B)$ where: $\Sigma$ denotes the set of possible labels of $\Pi$; $\C_W$ denotes the white-node constraint, and let $d_W$ be the size of the multisets in it; $\C_B$ denotes the black-node constraint, and let let $d_B$ be the size of the multisets in it.
The problem $\Pi'=(\Sigma', C'_W, \C'_B)=\re(\Pi)$ is defined as follows.
\begin{itemize}
    \item Let us first define $\C'_B$. Let $S$ be the maximal set satisfying the following two properties: (i) for all $\{\L_1, \ldots, \L_{d_B}\} \in S$ it holds that, for all $i$, $\L_i \in 2^{\Sigma} \setminus \{\emptyset\}$; (ii) for all $(\ell_1, \ldots, \ell_{d_B}) \in \L_1 \times \ldots \times \L_{d_B}$ it holds that $\{\ell_1, \ldots, \ell_{d_B}\}$ is in $\C_B$.  Then, $\C'_B$ is defined as the set $S$ where we remove all configurations $\{\L_1, \ldots, \L_{d_B}\}$ such that there exists another configuration $\{\bar{\L}_1, \ldots, \bar{\L}_{d_B}\} \in S$ and a permutation $\phi$, such that $\L_i \subseteq \bar{\L}_{\phi(i)}$ for all $i$, and there exists at least one $i$ such that the inclusion is strict. The removed configurations are called \emph{non-maximal}, while the remaining ones are called \emph{maximal}.

    \item $\Sigma' \subseteq 2^{\Sigma}$ contains all the sets that appear at least once in some configuration in $C_B$.

    \item $C_W$ is defined as all configurations $\{\L_1,\ldots,\L_{d_W}\}$ such that the following two properties hold: (i) for all $i$, $\L_i\in\Sigma'$; (ii) there exists $(\ell_1,\ldots,\ell_{d_W}) \in \L_1 \times \ldots \times L_{d_W}$ such that $\{\ell_1,\ldots,\ell_{d_W}\}$ is in $C_W$.
\end{itemize}
We now define $\rere(\Pi)$.
Let $(\Sigma', C'_B, \C'_W) = \re((\Sigma,C_B,C_W))$. The problem $\rere(\Pi)$ is defined as $(\Sigma', C'_W, \C'_B)$. In other words, $\rere$ is defined similarly as in $\re$, but the role of the black and white constraints are reversed.

We start by proving a lemma that states how the complexity of $\Pi$ is related, in the Supported LOCAL model, with the complexities of $\re(\Pi)$ and $\rere(\Pi)$.
\begin{lemma}\label{lem:speedup}
    Let $\Pi = (\Sigma,C_W,C_B)$ be a problem in the black-white formalism, where $\Delta'$ denotes the size of the multisets in $C_W$ and $r'$ the size of the multisets in $C_B$. 
    Let $G$ be a bipartite $2$-colored support graph, and let $n$ denote the number of nodes of $G$, $\Delta$ the maximum degree of the white nodes of $G$, and $r$ the maximum degree of the black nodes of $G$.
    Let $\mathcal G'$ be the class of all subgraphs of $G$ such that for every $G' \in \mathcal G'$, each white node of $G'$ has degree at most $\Delta'$ and each black node of $G'$ has degree at most $r'$. 
    Let $T := T(n,\Delta,r,\Delta',r') \geq 1$ be an integer such that
    \begin{enumerate}
        \item the girth of $G$ is at least $2T + 4$, and
        \item there exists a deterministic white (resp.\ black) algorithm $\mathcal A$ that solves $\Pi$ in $T$ rounds for each input graph $G' \in \mathcal G'$ (and support graph $G$).
    \end{enumerate}
    Then, there also exists a deterministic black (resp.\ white) algorithm $\mathcal A^*$ that solves $\Pi' := \re(\Pi)$ (resp.\ $\Pi' := \rere(\Pi)$) in $T - 1$ rounds for each input graph $G' \in \mathcal G'$ (and support graph $G$).
\end{lemma}
\begin{proof}
    Due to symmetry, it suffices to prove the lemma for the case that $\mathcal A$ is a white algorithm, $\mathcal A^*$ a black algorithm, and $\Pi' = \re(\Pi)$.
    Let $\mathcal A$ be a deterministic white algorithm satisfying the properties stated in the lemma.
    We start by defining $\mathcal A^*$, and then prove that it indeed solves $\Pi'$ in the required runtime.

    Consider a graph $G' \in \mathcal G'$ and a black node $v$ of $G'$, and let $e = \{ v, w \}$ be an edge of $G'$ incident to $v$.
    Let
    $Z_{T-1}(v)$ denote the subgraph of $G$ induced by all nodes in distance at most $T - 1$ from $v$.

    Let $\mathcal G^* \subseteq \mathcal G'$ denote the class of all graphs $G^* \in \mathcal G'$ containing $v$ for which the information stored in the nodes of $Z_{T-1}(v)$ (in particular regarding which edges are part of the input graph) is identical in $G^*$ and $G'$ (which in particular implies that $G^*$ contains $e$ and $w$).
    Let $\mathcal L_e \subseteq \Sigma$ denote the set of all output labels $L$ such that there exists a graph $G^* \in \mathcal G^*$ such that $w$ outputs $L$ on edge $e$ when executing $\mathcal A$ with input graph $G^*$.

    Let $e_1, \dots, e_y$ denote the edges incident to $v$ in $G'$.
    We now define the output of $v$ on the edges $e_i$ according to $\mathcal A^*$ as follows.
    Let $\mathcal L^*_{e_1}, \dots, \mathcal L^*_{e_y}$ denote an arbitrary sequence of $y$ subsets of $\Sigma$ satisfying
    \begin{enumerate}
        \item \label{prop1} $\mathcal L^*_{e_i} \supseteq \mathcal L_{e_i}$ for each $1 \leq i \leq y$,
        \item \label{prop2} for any choice $(L_1, \dots, L_y) \in \mathcal L^*_{e_1} \times \dots \times \mathcal L^*_{e_y}$, we have $\{ L_1, \dots, L_y \} \in C_B$, and 
        \item \label{prop3} for any sequence $\mathcal L'_{e_1}, \dots, \mathcal L'_{e_y}$ satisfying $\mathcal L^*_{e_i} \subseteq \mathcal L'_{e_i}$ for all $1 \leq i \leq y$ and $\mathcal L^*_{e_i} \not\subseteq \mathcal L'_{e_i}$ for at least one $1 \leq i \leq y$, there exists a choice $(L_1, \dots, L_y) \in \mathcal L'_{e_1} \times \dots \times \mathcal L'_{e_y}$ such that $\{ L_1, \dots, L_y \} \notin C_B$.
    \end{enumerate}
    (If no such sequence $\mathcal L^*_{e_1}, \dots, \mathcal L^*_{e_y}$ exists, we can assume that $\mathcal L^*_{e_i}$ is defined as $\mathcal L_{e_i}$ for each $1 \leq i \leq y$, but we will see that this cannot happen.)
    Now, $\mathcal A^*$ is defined so that $v$ outputs $\mathcal L^*_{e_i}$ on $e_i$ for each $1 \leq i \leq y$.
    This concludes the definition of $\mathcal A^*$.

    From the definition of $\mathcal A^*$, it is immediate that in order to compute the output for each incident edge (according to $\mathcal A^*$), it suffices for a black node $v$ to collect all information contained in $Z_{T-1}(v)$.
    Hence, the runtime of $\mathcal A^*$ is $T - 1$, as desired.
    It remains to show that $\mathcal A^*$ produces a correct solution for $\Pi' = (\Sigma', C'_W, C'_B)$.

    We first show that for each white node $w$ of degree $\Delta'$, the multiset of labels that $\mathcal A^*$ outputs on $w$'s incident edges in $G'$ is contained in $C'_W$.
    By the definition of $\mathcal A^*$, for each edge $e$ of $G'$, the set $\mathcal L_e$ contains the label that $\mathcal A$ outputs on $e$ when executed with input graph $G'$.
    Now, if we select this label for each edge incident to $w$ in $G'$, then we obtain a multiset of labels that is contained in $C_W$ (by the correctness of $\mathcal A$), which implies that the multiset of labels that $\mathcal A^*$ outputs on $w$'s incident edges in $G'$ is indeed contained in $C'_W$ (by the definition of $C'_W$).

    Now we prove that for each black node $v$ of degree $r'$, the multiset of labels that $\mathcal A^*$ outputs on $v$'s incident edges in $G'$ is contained in $C'_B$.
    Again, let $e_1, \dots, e_y$ denote the edges incident to $v$ in $G'$ (where $y = r'$), and, for each $1 \leq i \leq y$, let $w_i$ denote the endpoint of $e_1$ that is not $v$.
    
    We first show that for any choice $(L_1, \dots, L_y) \in \mathcal L_{e_1} \times \dots \times \mathcal L_{e_y}$, we have $\{ L_1, \dots, L_y \} \in C_B$.
    For a contradiction, assume that this is not true, and let $(L_1, \dots, L_y) \in \mathcal L_{e_1} \times \dots \times \mathcal L_{e_y}$ such that $\{ L_1, \dots, L_y \} \notin C_B$.
    By the definition of $\mathcal A^*$, it follows that for each $1 \leq i \leq k$, there exists a graph $G'_i$ such that
    \begin{enumerate}
        \item $Z_{T-1}(v)$ is identical in $G'$ and $G'_i$, and
        \item $w_i$ outputs $L_i$ on $e_i$ when executing $\mathcal A$ with input graph $G'_i$.
    \end{enumerate}
    Recall that in $\mathcal A$, each white node $w_i$ decides on the output on $e_i$ based solely on the information contained in the subgraph $Z_T(w_i)$ of $G$ induced by all nodes in distance at most $T$ from $w_i$.
    For each $1 \leq i \leq y$, let $\ext_i$ denote the graph induced by the nodes of $Z_T(w_i)$ that are not contained in $Z_{T-1}(v)$.
    Observe that, due to the fact that the girth of $G$ is at least $2T + 4$, the subgraphs $\ext_i$ are pairwise disjoint and nonadjacent (i.e., for $i \neq j$, no node of $\ext_i$ is identical or adjacent to a node from $\ext_j$).
    Hence, by the definition of the graph class $\mathcal G'$, there is a graph $\hat{G} \in \mathcal G'$ such that, for each $1 \leq i \leq y$, $Z_T(w_i)$ is identical in $\hat{G}$ and $G'_i$.
    It follows that, for each $1 \leq i \leq y$, node $w_i$ outputs $L_i$ on $e_i$ when executing $\mathcal A$ with input graph $\hat{G}$.
    However, this yields a contradiction to the correctness of $\mathcal A$ as the configuration $\{ L_1, \dots, L_y \}$ produced on the edges of $\hat{G}$ incident to $v$ is not contained in $C_B$.
    Thus, we obtain that for any choice $(L_1, \dots, L_y) \in \mathcal L_{e_1} \times \dots \times \mathcal L_{e_y}$, we have $\{ L_1, \dots, L_y \} \in C_B$.

    This in particular implies that there exist the subsets $\mathcal L^*_{e_1}, \dots, \mathcal L^*_{e_y}$ satisfying properties (\ref{prop1}) to (\ref{prop3}) as defined during the construction of $\mathcal A^*$.
    Moreover, properties (\ref{prop1}) and (\ref{prop3}) imply that the configuration $\{ \mathcal L^*_{e_1}, \dots, \mathcal L^*_{e_y} \}$ is actually a \emph{maximal} configuration (by the definition of maximality).
    Hence, $\{ \mathcal L^*_{e_1}, \dots, \mathcal L^*_{e_y} \} \in C'_B$, by the definition of $C'_B$.
    We conclude that $\mathcal A^*$ produces a correct solution for $\Pi' = (\Sigma', C'_W, C'_B)$.
\end{proof}

We now show what can be obtained by applying \Cref{lem:speedup} multiple times. In the following, an algorithm that is able to solve $\Pi$ on $(G,\mathcal G')$ in time $T$ denotes an algorithm that solves $\Pi$ in $T$ rounds on all the input graphs $G'' \in \mathcal{G'}$ when the support graph is $G$.
\begin{theorem}\label{lem:re-works}
    Let $\Pi = (\Sigma,C_W,C_B)$ be a problem in the black-white formalism, where $\Delta'$ denotes the size of the multisets in $C_W$ and $r'$ the size of the multisets in $C_B$. 
    Let $G$ be a bipartite $2$-colored support graph, and let
    $g$ denote the girth of $G$.
    Let $\mathcal G'$ be the class of all subgraphs of $G$ such that for every $G' \in \mathcal G'$, each white node of $G'$ has degree at most $\Delta'$ and each black node of $G'$ has degree at most $r'$. 
    Let $\Pi =: \Pi_0, \Pi_1, \dots, \Pi_k$ be a lower bound sequence.
    Assume that there is no deterministic $0$-round white algorithm that bipartitely solves $\Pi_k$ on
    $(G, \mathcal G')$.
    Then, bipartitely solving $\Pi$ with a deterministic white algorithm requires $\min\{2k,\frac{g-4}{2}\}$ rounds
    on $(G, \mathcal G')$.
\end{theorem}
\begin{proof}
    Assume for a contradiction that there is a deterministic white algorithm solving $\Pi$ on $(G, \mathcal G')$ in $\min\{2k,\frac{g-4}{2}\} - 1$ rounds.
    If $\min\{2k,\frac{g-4}{2}\} - 1$ is even, let $\mathcal A$ be such an algorithm, otherwise let $\mathcal A$ be a deterministic white algorithm solving $\Pi$ on $(G, \mathcal G')$ in $\min\{2k,\frac{g-4}{2}\}$ rounds.
    In either case, the runtime $T$ of $\mathcal A$ is even; hence, by applying \Cref{lem:speedup} iteratively $T$ times, we obtain that there is a white $0$-round algorithm solving $\Pi_{T/2}$ on $(G, \mathcal G')$.
    (Note that since the runtimes of the algorithms considered in the iterative applications of \Cref{lem:speedup} decrease with each iteration and $2 \min\{2k,\frac{g-4}{2}\} + 2 \leq g$, the premises of \Cref{lem:speedup} are satisfied in each iteration.)

    If $T = 2k$, we obtain a contradiction to the assumption in the theorem that there is no $0$-round white algorithm that solves $\Pi_k$ on $(G, \mathcal G')$.
    Hence, assume $T \neq 2k$, which implies $T < 2k$, by the definition of $T$.
    As we obtained that there is a white $0$-round algorithm solving $\Pi_{T/2}$ on $(G, \mathcal G')$, trivially there exists also a white $(2k - T)$-round algorithm $\mathcal A'$ solving $\Pi_{T/2}$ on $(G, \mathcal G')$ (and, as $T$ is even, also $2k - T$ is even).
    By applying \Cref{lem:speedup} another $2k - T$ times (starting from $\Pi_{T/2}$ and $\mathcal A'$), we obtain that there is a white $0$-round algorithm solving solving $\Pi_{k}$ on $(G, \mathcal G')$, yielding the desired contradiction.
\end{proof}

In the following, by girth of a hypergraph $G$ we denote half of the girth of the incidence graph of $G$. Recall that the rank of a hyperedge $e$ is $|e|$.
\begin{corollary}\label{cor:re-works-hyper}
    Let $\Pi = (\Sigma,C_W,C_B)$ be a problem in the black-white formalism, where $\Delta'$ denotes the size of the multisets in $C_W$ and $r'$ the size of the multisets in $C_B$. 
    Let $G$ be a support hypergraph, and let
    $g$ denote the girth of $G$.
    Let $\mathcal G'$ be the class of all subhypergraphs of $G$ such that for every $G' \in \mathcal G'$, each node of $G'$ has degree at most $\Delta'$ and each hyperedge of $G'$ has rank at most $r'$.
    Let $\Pi =: \Pi_0, \Pi_1, \dots, \Pi_k$ be a lower bound sequence.
    Assume that there is no deterministic $0$-round algorithm that non-bipartitely solves $\Pi_k$ on
    $(G, \mathcal G')$.
    Then, non-bipartitely solving $\Pi$ with deterministically requires $\min\{k,\frac{g-4}{2}\}$ rounds
    on $(G, \mathcal G')$.
\end{corollary}
\begin{proof}
    This directly follows from \Cref{lem:re-works} from the usual equivalence between hypergraphs and bipartite $2$-colored graphs, where a node, resp.\ hyperedge, in the hypergraph corresponds to a white node, resp.\ black node, in the bipartite graph and there is an edge between a white node $w$ and a black node $v$ in the bipartite graph if and only if the hypergraph node corresponding to $w$ is contained in the hyperedge corresponding to $v$.
\end{proof}

\section{Randomized Lower Bounds}\label{sec:rand}
In the LOCAL model, it is known that, for all problems $\Pi$ belonging to a large family of problems $\Pi$ called \emph{component-wise verifiable} (which includes all problems in the black-white formalism), the deterministic and randomized complexity of $\Pi$ cannot differ by too much, and in particular that the deterministic complexity of $\Pi$ on instances of size $n$ is at most the randomized complexity of $\Pi$ on instances of size $2^{n^2}$. This result was first proved for randomized algorithms using an amount of random bits that is bounded as a function of $n$ and $\Delta$ \cite{ChangKP19}, and then extended to all randomized algorithms \cite{derandomization}.
We show that such statements hold in the Supported LOCAL model as well. We first state the result that is known in the LOCAL model, adapted to our restricted setting of problems in the black-white formalism.
\begin{lemma}[\cite{ChangKP19,derandomization}, rephrased]\label{lem:derand}
    Let $\Pi$ be a problem in the black-white formalism. Let $D_\Pi(n,c)$ be the deterministic complexity of $\Pi$ in the LOCAL model, when the ID space is $\{1,\ldots,n^c\}$, and let  $R_\Pi(n)$ be the randomized complexity of $\Pi$ in the LOCAL model, for algorithms with failure probability bounded by $1/n$.
    Let $\mathcal{G}_{n,c}$ be the set of possible instances of size $n$ with ID assignments from $\{1,\ldots,n^c\}$. Then, 
    \[
        D_\Pi(n,c) \le R_\Pi(|\mathcal{G}_{n,c}| + 1).
    \]
\end{lemma}
\noindent In \cite{ChangKP19,derandomization}, $|\mathcal{G}_{n,c}|$ is bounded as follows:
\begin{itemize}
    \item The number of possible graphs of size $n$ are at most $2^{n \choose 2}$ (we can describe each graph by a bit-string representing whether an edge is present or not, and there are at most $n \choose 2$ edges).
    \item The number of possible ID assignments in $G$ are $2^{c n \log n}$ for some constant $c$, since to each node we assign an ID in $\{1,\ldots,n^c\}$.
    \item Assume each node receives a reasonably small input, say of $O(1)$ bits, then there are at most $2^{O(n)}$ possible input assignments.
\end{itemize}
Thus, the number of possible instances is bounded by $2^{{n \choose 2} + c n \log n + O(n)}$ which, for large enough $n$, is strictly less than $2^{n^2}$. We thus get that, in the LOCAL model, $D_\Pi(n,c) \le R_\Pi(2^{n^2})$. On a high level, in \cite{ChangKP19}, the statement of \Cref{lem:derand} is proved as follows.
\begin{itemize}
    \item The randomized algorithm is simulated on all graphs $G \in \mathcal{G}_{n,c}$, by telling it that there are $2^{n^2}$ nodes. The algorithm cannot detect the lie, since this situation is indistinguishable from the case in which the algorithm is run on a larger graph (of exactly $2^{n^2}$ nodes) in which there is a connected component that is $G$.
    \item In such a simulation, the failure probability of the algorithm is at most $1 / 2^{n^2}$.
    \item By a union bound argument, there exists a function $f$ that maps IDs into bit strings such that, for all $G \in \mathcal{G}_{n,c}$, if we run the randomized algorithm on $G$ by using as random bits the bits that are (deterministically) given by $f$, the algorithm succeeds on all nodes of $G$. 
    \item A deterministic algorithm is obtained by running the randomized algorithm, where, instead of using random bits, nodes use the bits given by $f$.
\end{itemize}
This proof requires to have some bound on the amount of random bits used by the algorithm, and~\cite{derandomization} shows how to remove this assumption.

While the Supported LOCAL model is strictly stronger than the LOCAL model, we can also see the Supported LOCAL model as a special case of LOCAL in which the input given to the nodes satisfies some special strong properties, in particular:
\begin{itemize}
    \item All nodes know the graph $G$;
    \item Each node receives some additional input, that is, which of its incident edges are part of the subgraph $G'$.
\end{itemize}
Moreover, we note that lying about the size of the graph in the Supported LOCAL model is still possible. In fact, suppose that the support graph $G$ has size $n$, and that we want to run an algorithm by telling it that the number of nodes is $N > n$. Nodes can consistently imagine that the support graph contains two components: one is $G$, and the other is an arbitrary graph of size $N - n$.
Hence, \Cref{lem:derand} directly works in the Supported LOCAL model as well, but we need to bound the possible instances differently. We bound the possible instances as follows.
\begin{itemize}
    \item The number of possible graphs of size $n$ are at most $2^{n \choose 2}$.
    \item The number of possible ID assignments in $G$ are $2^{c n \log n}$ for some constant $c$. However, since all nodes know $G$, they can recompute a new ID assignment over the IDs $\{1,\ldots,n\}$. Hence, we can w.l.o.g.\ assume that the possible ID assigments are just $n! \le 2^{n \log n}$.
    \item Each edge is marked to specify whether it is part of the input graph or not. For this purpose, $1$ bit of information per edge is sufficient, and hence there are at most $2^{n^2}$ input assignments.
\end{itemize}
We thus get that the number of possible instances is bounded by $2^{n \choose 2} \cdot 2^{n \log n} \cdot 2^{n^2} \le 2^{3n^2}$. Hence, we obtain the following.
\begin{lemma}\label{lem:derand-supported}
    Let $\Pi$ be a problem in the black-white formalism. Let $D_\Pi(n)$ be the deterministic complexity of $\Pi$ in the Supported LOCAL model (which, w.l.o.g., does not depend on the size of the ID space), and let $R_\Pi(n)$ be the randomized complexity of $\Pi$ in the Supported LOCAL model, for algorithms with failure probability at most $1/n$.
    Then, 
    \[
        D_\Pi(n) \le R_\Pi(2^{3n^2}).
    \]
\end{lemma}
\noindent We show that a similar statement holds for hypergraphs as well. Recall that a hypergraph is linear if each pair of hyperedges share at most one node. Consider the case in which each hyperedge has size at least $2$. Since hyperedges share at most one node, each node can be incident to at most $n-1$ hyperedges. Thus, there are at most $n^2$ hyperedges in total. This implies that we can describe all possible linear hypergraphs (where each hyperedge has size at least $2$) with $2^{2 n^2 \lceil \log n \rceil}$ bits, by using, for each node, $2 n \lceil \log n \rceil$ bits to describe an array of $n-1$ integers in $\{1,\ldots,n^2\}$. Then, similarly as in the case of graphs, the possible ID assignments for the nodes are $2^{n \log n}$, and there are at most $2^{n^3}$ input assignments for the possible node-hyperedge pairs in order to describe the input graph. Thus, we obtain the following.
\begin{theorem}\label{lem:hyperderandomization}
    Let $\Pi$ be a problem in the black-white formalism. Let $D_\Pi(n)$ be the deterministic complexity of $\Pi$ in the Supported LOCAL model (which, w.l.o.g., does not depend on the size of the ID space), and let $R_\Pi(n,\Delta,r)$ be the randomized complexity of $\Pi$ in the Supported LOCAL model, for algorithms with failure probability at most $1/n$. Then, on linear hypergraphs where each hyperedge has size at least $2$,
    \[
        D_\Pi(n) \le R_\Pi(2^{4n^3}).
    \]
\end{theorem}

\end{document}